%% file: I3E_transactions.tex
\documentclass[12pt]{article}
\usepackage{fullpage}
\usepackage[cmex10]{amsmath}
\usepackage{amsfonts}%
\usepackage{amssymb}%
\usepackage{amsthm}

\usepackage{graphicx}
\usepackage{epstopdf}
\usepackage{psfrag}
\usepackage{tipa}
\usepackage{algorithmic}
\usepackage{array}
\usepackage{mdwmath}
\usepackage{mdwtab}
\usepackage{eqparbox}
\usepackage{subfigure}
\usepackage{bbm}
\newtheorem{theorem}{Theorem}

\newtheorem{proposition}{Proposition}
\newtheorem{corollary}{Corollary}

\newtheorem{claim}{Claim}

\theoremstyle{remark}
\newtheorem*{remark}{Remark}

\newcommand{\R}{\mathbb{R}}

\newcommand{\bH}{\mathbf{H}}

\begin{document}
\title{Opportunistic Scheduling in Heterogeneous Networks: Distributed Algorithms and System Capacity}
\author{
Joseph Kampeas, Asaf Cohen and Omer Gurewitz
\\
Department of Communication Systems Engineering
\\
Ben-Gurion University of the Negev
\\
Beer-Sheva, 84105, Israel
\\ {\tt Email: \{kampeas,coasaf,gurewitz\}@bgu.ac.il}}
\date{}
\maketitle
\begin{abstract}
In this work, we design and analyze novel distributed scheduling algorithms for multi-user MIMO systems. In particular, we consider algorithms which do not require sending channel state information to a central processing unit, nor do they require communication between the users themselves, yet, we prove their performance closely approximates that of a centrally-controlled system, which is able to schedule the strongest user in each time-slot.

Our analysis is based on a novel application of the Point-Process approximation. This novel technique allows us to examine non-homogeneous cases, such as non-identically distributed users, or handling various QoS considerations, and give exact expressions for the capacity of the system under these schemes, solving analytically problems  which to date had been open.
Possible application include, but are not limited to, modern 4G networks such as 3GPP LTE, or random access protocols.
\end{abstract}

\input{introduction_trans}

\input{prelim_trans}
\input{dist_alg_trans}

\input{hetero_trans}
\input{capture_trans}

\input{enhance_trans}

\input{con_trans}
\input{app_A_trans}
\input{app_C_trans}

% Add how to derive from point of process the  block maxima analysis

\bibliographystyle{IEEEtran}

\bibliography{bibliography}

\end{document}

%% file: introduction_trans.tex
\newcommand{\commenta}[1]{{\bfseries Asaf says: #1}}
\section{Introduction}
Consider the problem of scheduling users in a multi-user MIMO system. For several decades, at the heart of such systems stood a basic division principle: either through TDMA, FDMA or more complex schemes, users did not use the medium jointly, but rather used some scheduling mechanism to ensure only a single user is active at any given time. Numerous medium access (MAC) schemes at the data link layer also, in a sense, fall under this category. Modern multi-user schemes, such as practical multiple access channel codes or dirty paper coding (DPC) for Gaussian broadcast channels \cite{weingarten2006capacity}, do allow concurrent use of a shared medium, yet, to date, are complex to implement in their full generality. As a result, even modern 4G networks consider scheduling \emph{groups of users}, each of which employing a complex multi-user code \cite{sesia2009lte,IEEE_802.16m}.

Hence, scheduled designs, in which only a single user or a group of users utilize the medium at any given time, are favorable in numerous practical situations. In these cases, the goal is to design an efficient schedule protocol, and compute the resulting system capacity. In this work, we derive the capacity of multi-users MIMO systems under distributed scheduling algorithms, in which each user experiences a different channel distribution, subject to various QoS considerations.
%%%%%%%%%%%%%%%%%%%%%%%%%%%%%
\subsection{Related Work}
Various suggested protocols in the current literature follow the pioneering work of \cite{knopp1995information}. In these systems, at the beginning of a time-slot, a user computes the key parameters relevant for that time-slot. For example, the channel matrix $\bH$ (Figure \ref{MU-MIMO}). It then sends these parameters to a central processing unit, which decides which user to schedule for that time-slot. This enables the central unit to optimize some criterion, e.g., the number of bits transmitted in each slot, by scheduling the user with the best channel matrix. This is the essence of \emph{multi-user diversity}. In \cite{yoo2006optimality}, the authors adopted a zero-forcing beamforming strategy, where users are selected to reduce the mutual interference. The scheme was shown to asymptotically achieve the performance of DPC. An enhanced cooperative scheme, in which base stations optimize their beamforming coordination,  such that the users' transmitting power is subject to SINR  minmax fairness, is given in \cite{zakhour2011minmax}.
Its analysis showed that optimal beamforming strategies have an equivalent convex optimization problem. Yet, its solution requires centralized
CSI knowledge.
In \cite{chen2006enhancing}, the authors devised a  multi-user diversity with interference avoidance by  mitigation approaches, which selects the user with the highest minimal eigenvalue of his Wishart channel matrix $\textbf{H}\textbf{H}^\dagger$.
\cite{jagannathan2007scheduling} proposed a scheduling scheme that transmit only to a small subset of heterogeneous users with favorable channel characteristics. This provided near-optimal performance when the total number of users to choose from was large. Scaling laws for the sum-rate capacity comparing maximal user scheduling, DPC and BF were given in \cite{sharif2007comparison}. Additional surveys can be found in \cite{caire2006mimo,hassibi2007fundamental}. Subsequently, \cite{choi2008capacity} analyzed the scaling laws of maximal \emph{base station} scheduling via Extreme Value Theory (EVT), and showed that by scheduling the station with the strongest channel among $K$ stations (Figure \ref{MU-MIMO_stations}), one can gain a factor of $O(\sqrt{2\log K})$ in the expected capacity compared to random or Round-Robin scheduling.

Extreme value theory and order statistics are indeed the key methods in analyzing the capacity of such scheduled systems. In \cite{wang2007coverage}, the authors suggested a subcarrier assignment algorithm (in OFDM-based systems), and used order statistics to derive an expression for the resulting link outage probability. Order statistics is required, as one wishes to get a handle on the distribution of the \emph{selected} users, rather than the a-priori distribution. In \cite{pun2007opportunistic}, the authors used EVT to derive throughput and scaling laws for scheduling systems using beamforming and various linear combining techniques. \cite{choi2010user} discussed various user selection methods in several MIMO detection schemes. The paper further strengthened the fact that appropriate user selection is essential, and in several cases can even achieve optimality with sub-optimal detectors. Additional user-selection works can be found in \cite{airy2003spatially,swannack2004low,primolevo2005channel,yoo2006finite}.

In \cite{qin2003exploiting,qin2006distributed}, the authors suggested a decentralized MAC protocol for OFDMA channels, where each user estimates his channels gain and compares it to a threshold. The optimal threshold is achieved when only one user exceeds the threshold on average. This distributed scheme achieves $1/e$ of the capacity which could have been achieved by scheduling the strongest user. The loss is due to the channel contention inherent in the ALOHA protocol.
\cite{bai2006opportunistic} extended the distributed threshold scheme for multi-channel setup, where each user competes on $m$ channels. In \cite{qin2008distributed} the authors used a similar approach for power allocation in the multi-channel setup, and suggested an algorithm that asymptotically achieves the optimal water filling solution. To reduce the channel contention, \cite{qin2003exploiting,qin2004opportunistic} introduced a splitting algorithm which resolves collision by allocating several mini-slots devoted to finding the best user. Assuming all users are equipped with a collision detection (CD) mechanism,  the authors also analyzed the suggested protocol for users that are not fully backlogged, where the packets randomly arrive with a total arrival rate $\lambda$ and for channels with memory. \cite{to2010exploiting} used a similar splitting approach to exploit idle channels in a multichannel setup, and showed improvement of $63\%$ compared to the original scheme in \cite{qin2003exploiting}.
%%%%%%%%%%%%%%%%%%%%%%%%%%%%%%%%%%%

\subsection{Main Contribution}
In this work, we suggest a novel technique, based on the Point Process approximation, to analyze the expected capacity of scheduled multi-user MIMO systems. We first briefly show how this approximation allows us to derive recent results described above. However, the strength of this approximation is in facilitating the asymptotic (in the number of users) analysis of the capacity of such systems in \emph{different non-uniform scenarios, where users are either inherently non-uniform or a forced to act this way} due to Quality of Service constrains. We compute the asymptotic capacity for non-uniform users, when users have un-equal shares or when fairness considerations are added. To date, these scenarios did not yield to rigorous analysis.

Furthermore, we suggest a novel distributed algorithm, which achieves a constant factor of the maximal multi-user diversity without centralized processing or communication among the users. Moreover, we offer a collision avoidance enhancement to our algorithm, which asymptotically achieves the maximal multi-user diversity without any collision detection mechanism.

The rest of this paper in organized as follows. In Section \ref{sec. prelim}, we describe the system model and related results. In Section \ref{sec. dist}, we describe the Point of Process technique and briefly show how it is utilized. In Section \ref{sec. hetero} we analyze the non-uniform scenario. In Section \ref{sec. capture} we examine the expected capacity in a non-uniform environment, assuming that the receiver can recover the message from a single collision. In Section \ref{sec. enhance} we describe the distributed algorithm and analyze its performance. Section \ref{sec. conc} concludes the paper.

\begin{figure}[t]
\centering
\subfigure[]{
\includegraphics[scale=0.3]{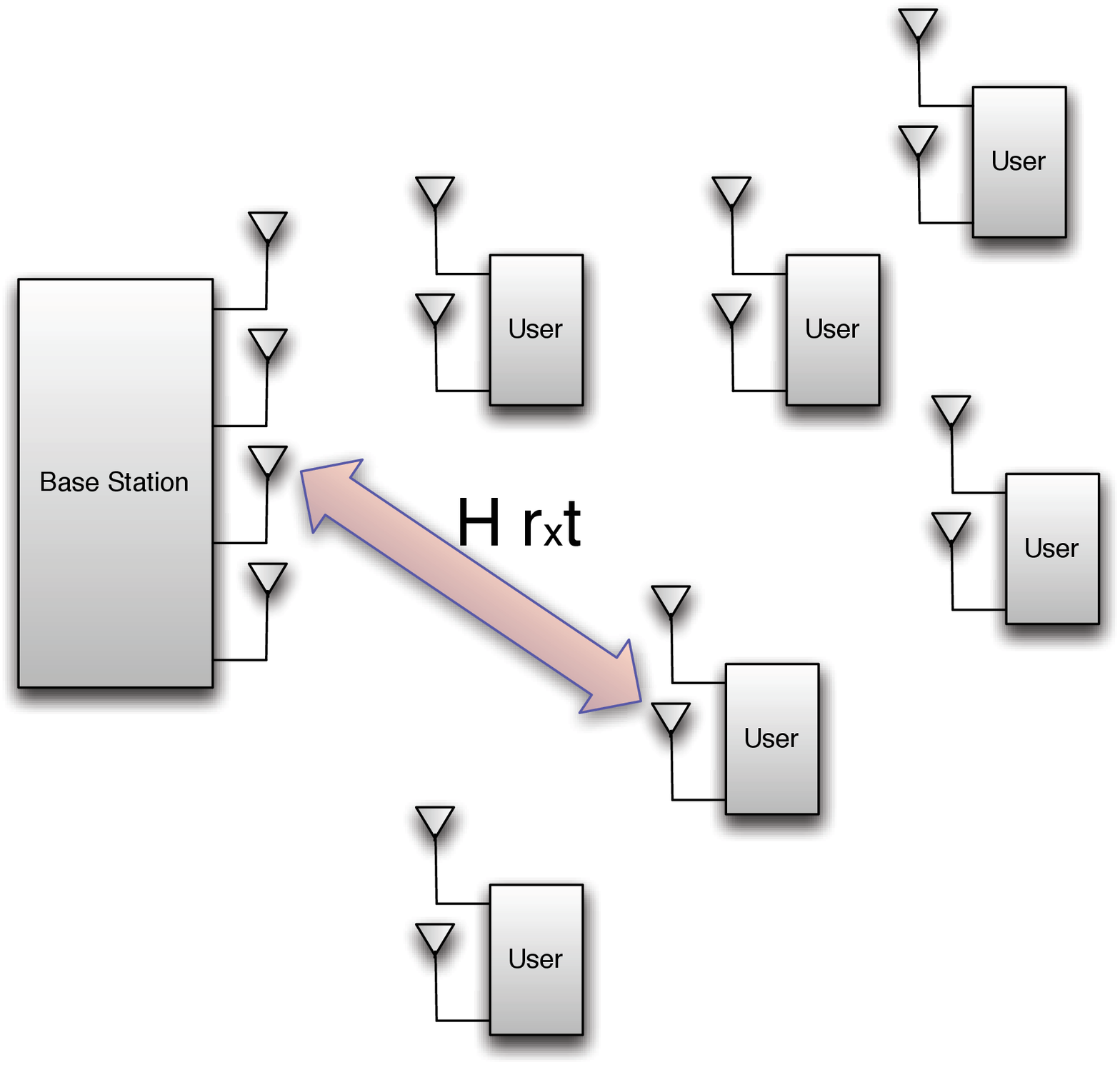}
\label{MU-MIMO}
}
\subfigure[]{
\includegraphics[scale=0.3]{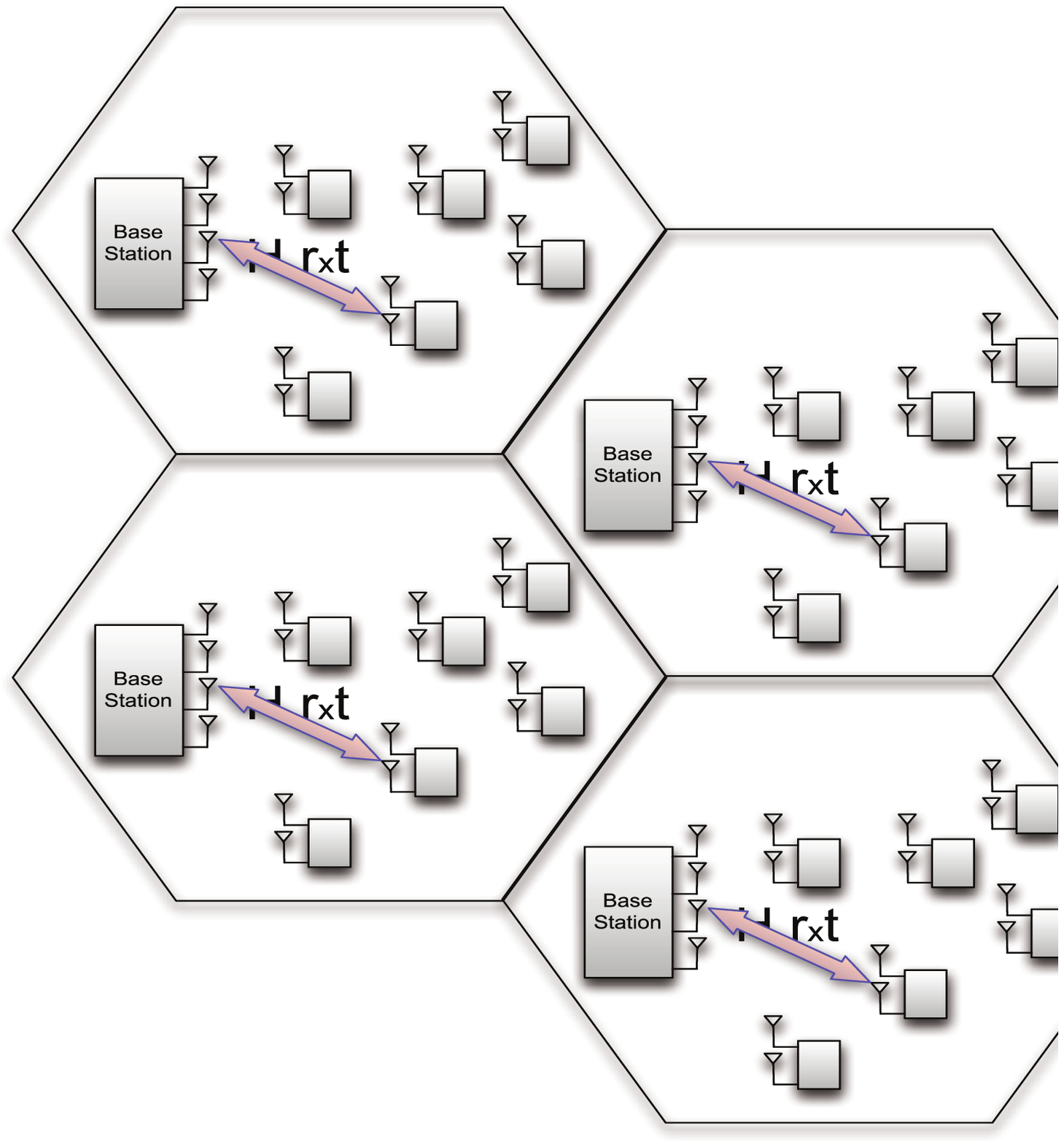}
\label{MU-MIMO_stations}
}
%\subfigure[]{
%\includegraphics[scale=0.30]{capacity1.pdf}
%\label{fig: capacity distribution}
%}
\caption[]{(a) Multi-user MIMO. (b) MU-MIMO stations. }
\end{figure}

%% file: prelim_trans.tex
\section{Preliminaries}\label{sec. prelim}
We consider a multiple-access model with $K$ users. The channel model is the following: $$\textbf{y} = \textbf{H}x + \textbf{n}$$
where $\textbf{y}\in\mathbb{C}^r$ is  the received vector and  $r$  is the number of receiving antennas.
$x\in\mathbb{C}^t$ is the transmitted vector constrained in its total power to $P$, i.e., $E[x^\dag x]\leq P$, where $t$ is the number of transmitting antennas.
$\textbf{H}\in\mathbb{C}^{r\times t}$ is a complex random Gaussian channel matrix such that all the entries are random $i.i.d.$ complex Gaussian with independent imaginary and real parts, zero mean and variance $1/2$ each.
$\textbf{n}\in \mathbb{C}^r$ is uncorrelated complex Gaussian noise with independent real and imaginary parts, zero mean and variance 1.
In the MIMO uplink model, we assume that the channel $\textbf{H}$ is known at the transmitters. In the centralized scheme,  the transmitters send their channel statistics to the receiver. I.e., the channel output at the receiver consist of the pair $(\textbf{y},\textbf{H})$. Then, the receiver lets to the transmitter with the strongest channel to transmit in the next slot.
In the MIMO downlink model, we assume that the channel $\textbf{H}$ is known at the receivers. In the centralized scheme, the receivers send their channel statistics to the transmitter, so he can choose the receiver that will benefit most from his transmission.
Moreover, we assume that the channel is memoryless, such that for each channel use, an independent realization of $\textbf{H}$ is drawn.
Through this paper, we use bold face notation for random variables.
% we denote $m = min(r,t)$ and $n=max(r,t)$ and use bold face notation for random variables.

%In this section, we derive the capacity gain of several opportunistic scheduling algorithms where a transmission in each time slot is allocated only to one user that is able to transmit at the highest rate among users.
%In order to derive the capacity gain of the following scheme, we analyze the asymptotic behavior of the maximum of $n$ $i.i.d.$ random variables, when $n$ is sufficiently large, using known Extreme value theorems.\\
%Extreme value theorem is a branch in statistics describing the unusual rather than the usual, i.e., maximum and minimum values of a random sample. The analysis of extreme value ,in the context of MIMO channels, was first investigated by \cite{choi2008capacity}.
%We would like to analyze the expected sum rate capacity when in each time slot we select the user with the largest capacity among all users, and observe the benefits of doing so. \\
\subsection{ MIMO Capacity }
\cite{smith2002gaussian, girko1997refinement} and \cite{chiani2003capacity} show that when the  elements of the channel gain matrix, $\textbf{H}$, are i.i.d.\ zero mean with finite moments up to order $4 + \delta$, for some $\delta > 0$ then the distribution of the capacity follows the Gaussian distribution by the CLT, as we can see in Figure~\ref{fig: capacity distribution}, with mean that grows linearly with $min(r,t)$, and variance which is mainly influenced by the power constraint $P$.

%It is convenient, for practical reasons, to work with a single user at each time, and we examine possible ways to select him.
With the observation that the channel capacity follows the Gaussian distribution, we would first like to investigate the extreme value distribution that the capacity follows, and thus retrieve the capacity gain when letting a user with the best channel statistics among all other users, utilize a  slot.

\subsection{Extreme Value Analysis for the Maximal Value}
In this sub-section we review the Extreme Value Theorem (EVT), from  \cite{leadbetter1983},\cite{coles2001introduction} and \cite{eastoem453}, that will later be used for asymptotic capacity gain analysis.
%In implementing this model for user capacities dataset obtain size can be critical. The choice amounts to a trade-off between bias and variance. hat are too small mean that the approximation by the limit distribution is likely to be poor, leading to bias in estimation and extrapolation. On the other hand, large blocks generate few block maxima, leading to large estimation variance.
\begin{theorem}[\cite{de2006extreme,leadbetter1983,eastoem453}]
\mbox{}
\begin{enumerate}
	\item[(i)]
 Suppose that $\textbf{x}_1,..\textbf{x}_n$ is a sequence of $i.i.d$ random variables with distribution function $F(x)$, and let $$M_n = \max(\textbf{x}_1,...,\textbf{x}_n). $$
If there exist a sequence of normalizing constants $a_n>0$ and $b_n$ such that as $n\rightarrow \infty$,
\begin{equation}\label{eqn: P(Mn) approx}
\Pr(M_n \leq a_nx + b_n)\stackrel{i.d.}{\longrightarrow} G(x)
\end{equation}
for some non-degenerate distribution G, then G is of the generalized extreme value (GEV) distribution type
\begin{equation}\label{eqn: G def}
G(x) = \exp\left\{-(1+\xi x)^{-1/\xi}\right\}
\end{equation}
and we say that $F(x)$ is in the domain of attraction of $G$,
where $\xi$ is the shape parameter, determined by the ancestor distribution $F(x)$ with the following relation.
\item[(ii)]
Let $h$ be the following reciprocal hazard function
\begin{equation}\label{eqn: h(x)}
  h(x)=\frac{1-F(x)}{f(x)} \textmd{ for } x_F \leq x \leq x^F ,
\end{equation}
where $x_F = inf \{x:F(x)>0\}$ and $x^F = sup \{x: F(x)< 1\}$ are the lower and upper endpoints of the ancestor distribution, respectively.
Then the shape parameter $\xi$ is obtained as the following limit,
\begin{equation}\label{eqn: dh dx}
  \frac{d}{dx}h(x)\stackrel{x\rightarrow x^F}{\longrightarrow} \xi .
\end{equation}
\item[(iii)] %Theorem 1.5.3 @ Leadbetter
If $\{\textbf{x}_n\}$ is an i.i.d. standard normal sequence of random variables, then the asymptotic distribution of $M_n= \max (\textbf{x}_1,...\textbf{x}_n)$ is a Gumbel distribution. Specifically,
$$\Pr(M_n \leq a_n x + b_n) \longrightarrow e^{-e^{-x}} $$
where
\begin{equation}\label{eqn: a_n normalized}
   a_n = (2\log n)^{-1/2}
\end{equation}
and
\begin{equation}\label{eqn: b_n normalized}
  b_n = (2\log n)^{1/2}-\frac{1}{2}(2\log n)^{-1/2}(\log \log n + \log 4\pi).
\end{equation}
\end{enumerate}
\end{theorem}
In Figure~\ref{fig: max capacity distribution} we see the max value distribution for $500$ observations which following the Gaussian distributed simulated in Figure~\ref{fig: capacity distribution} 
For completeness, a sketch of the proof is given in Appendix~\ref{sec: appendix A}.
Similarly, if $\{\textbf{x}_n\}$ follows the Gaussian distribution with mean $\mu$ and variance $\sigma^2$,
then the above theorem normalizing constants results in
\begin{equation}\label{eqn: a_n with mu and sigma}
   a_n = \sigma (2\log n)^{-\frac{1}{2}}
\end{equation}
and
\begin{equation}\label{eqn: b_n with mu and sigma}
   b_n = \sigma \left[(2\log n )^{\frac{1}{2}}- \frac{1}{2}(2\log n)^{-\frac{1}{2}}[\log\log n + \log(4\pi)]\right] + \mu.
\end{equation}
It follows that  for a Gaussian distribution,
$$a_n = \sigma(2\log n)^{-\frac{1}{2}}\rightarrow0,$$
which implies that
\begin{equation}\label{eqn: Mn limit}
   M_n\sim b_n \sim \sigma (2\log n )^{\frac{1}{2}} + \mu.
\end{equation}
\begin{figure}
\centering
  \includegraphics[scale=0.75]{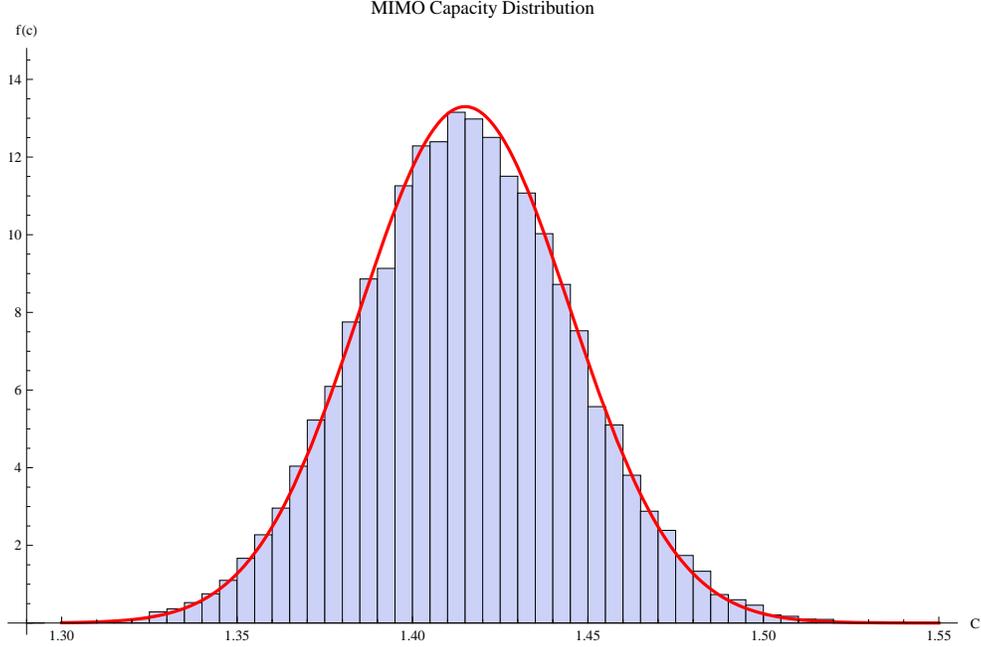}
  \caption{MIMO capacity distribution for  m=32 transmitting antennas and n=128 receiving antennas vs. Gaussian Distribution with $\mu = \sqrt{2}$ and $\sigma = 0.03$ (red line). }\label{fig: capacity distribution}
\end{figure}
\begin{figure}
\centering
  \includegraphics[scale=0.70]{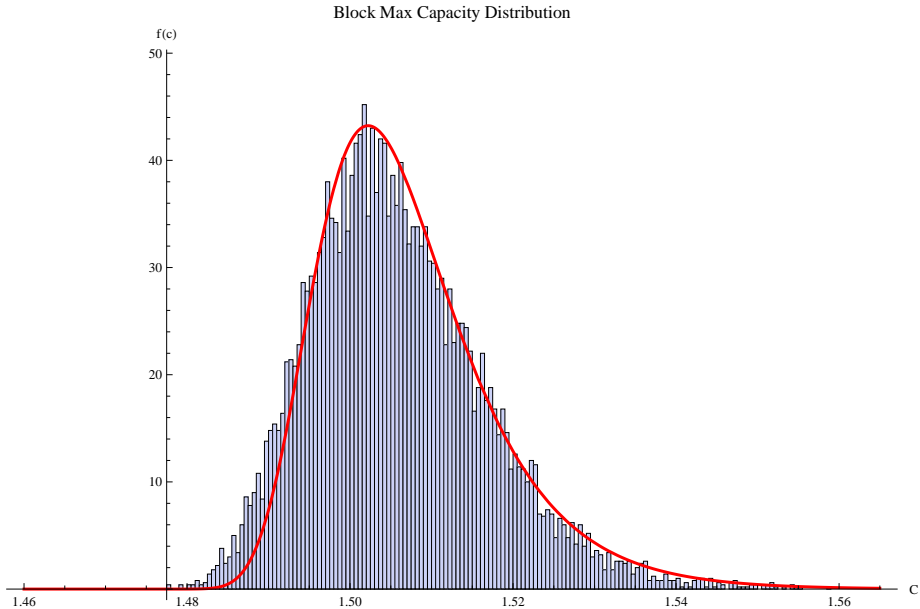}
  \caption{Maximal capacity distribution, when choosing the maximal capacity among 500 capacities that following the Gaussian distribution simulated  in Figure \ref{fig: capacity distribution}, with $\mu=\sqrt{2}$ and $\sigma = 0.03$. The red line is the corresponding Gumbel density plotted in range $[\mu + 2\sigma, \mu + 5\sigma].$ }\label{fig: max capacity distribution}
\end{figure}
\subsection{Multi-User Diversity}
   Assuming MIMO uplink model, i.e., perfect CSI of $K$ users at the receiver, then the expected capacity that we achieve by choosing the maximal user in each time slot will follow the expected value of Gumbel distribution with parameters $a_K,b_K$ \cite{choi2008capacity}, i.e.,
\begin{eqnarray}\label{eqn: Gumbel Expexted Value}
    E[M_K] &\stackrel{(a)}{=}& \sigma \left( b_K + a_K\gamma \right) + \mu  \\
    &\stackrel{(b)}{=}& \sigma \left[(2\log K )^{\frac{1}{2}}- \frac{1}{2}(2\log K)^{-\frac{1}{2}}[\log\log K + \log(4\pi)]  + \gamma (2\log K)^{-\frac{1}{2}}  \right] + \mu \nonumber
    %&\stackrel{(c)}{\rightarrow}& \sigma (2\log K )^{\frac{1}{2}} + \mu\nonumber
\end{eqnarray}
where $\gamma\approx 0.57721$ is Euler-Mascheroni constant, $(a)$ follows from the expectation of the Gumbel distribution and $(b)$ follows from (\ref{eqn: a_n normalized}) and (\ref{eqn: b_n normalized}). % respectively, and $(c)$ follows by using (\ref{eqn: Mn limit}) relation.
Hence, for large enough $K$,
$$E[M_K] = \sigma (2\log K )^{\frac{1}{2}} + \mu + o\left(\frac{1}{\sqrt{\log K}}\right)$$
That is, for large number of users, the expectation capacity grows like $\sqrt{2\log K}$. 

%% file: dist_alg_trans.tex
\section{ Distributed Algorithm }\label{sec. dist}
A major drawback of the previous method is that a base station must receive a perfect CSI from all users in order to decide which user is adequate to utilize the next time slot, which may not be feasible for a large number of users. Moreover, the delay caused by transmitting CSI to the base station would limit the performance.

In this section, we begin our discussion from a distributed algorithm, shown in \cite{bai2006opportunistic}, in which stations do not send their channel statistics to the base station, yet, with some subtle enhancements, the performance is asymptotically equal to that in (\ref{eqn: Gumbel Expexted Value}).
We provide an alternative analysis to this algorithm, that will serve us later in this paper.

The algorithm is as follows. Given the number of users, we set a high capacity threshold such that only a small fraction of the users will exceed it. In each slot, the users estimate their own capacity. If the capacity seen by a user is greater than the capacity threshold, he transmits in that slot. Otherwise, the user keeps silent in that slot. The base station can successfully receive the
transmission if no collision occurs.

%One should notice that if we were satisfied in one of the users with capacity that exceeded a certain threshold $u$, instead of choosing the user with maximal capacity, then we guaranteed an average capacity of at least $u$. Moreover, the user with largest capacity is among the users that their capacity exceeded threshold for certain.Since collisions may occur we obtain the following, the transmission is not
Let $C_{av}\left(u_{k}\right)$ denote the expected capacity, given a threshold $u_{k}$ such that $k\ll K$ i.i.d.\ users exceed it on average. For sufficiently large number of users, $K$, we obtain the following.
\begin{proposition}\label{prop: C av uniform users}
The expected capacity when working with a single user in each slot is
\begin{eqnarray}\label{eqn: threshold capacity throuput}
  % \left(1-\Pr(\textmd{unutilized slot})\right)E[C|C>u_{k/K}]\\
  C_{av}\left(u_{k}\right)  &=& k e^{-k}\left( u_{k} + \sigma a_K  \right) + o(a_{K})
\end{eqnarray}
where $a_K$ is the normalizing constant given in (\ref{eqn: a_n normalized}).
\end{proposition}
Due to the distributed nature of the algorithm, some slots will be \emph{idle} if no user exceeds the threshold, that is, no user transmits in that slot. Or, \emph{collisions} may occur if more than one user exceed the threshold, that is, more than one user is trying to transmit in a slot. Thus, we say that a slot is \emph{utilized} if exactly one user exceeds the threshold, namely, exactly one user transmits in a slot. Indeed, the expected capacity $C_{av}\left(u_{k}\right)$ has the form
$$C_{av}\left(u_{k}\right)= \Pr\left(\textmd{utilized slot}\right)E[C|C>u_{k}]$$
where
\begin{equation}\label{eqn: probability unutilized slot uniform users}
   \Pr\left(\textmd{utilized slot}\right) = k e^{-k}
\end{equation}
and
\begin{equation}\label{eqn: threshold expected capacity}
    E[C|C>u_{k}] =  u_{k} + a_{K} + o(a_{K}).
\end{equation}
That is, to compute $C_{av}\left(u_{k}\right)$ we analyze the expected capacity when letting a user with above-threshold-capacity utilize a slot, and the probability
that only a single user utilizes the slot.
%Hence, we need to analyze the expected capacity gained when letting a user with above threshold capacity to utilize a slot, and the probability that a single user utilizes the slot.\\
We choose to prove the above through the point process method \cite{eastoem453,smith1989extreme}. With the point of process, we can model and analyze the occurrence of large capacities, which  can be represented as a  point process, when considering the users index along with the capacity value. Later, in the main contribution of  this paper, this method will allow us to analyze the non-uniform case as well.

The following two subsections sketch the key steps to
prove Proposition~\ref{prop: C av uniform users}. The first discusses the estimation of the threshold, given the fraction of users which are required to pass it on average.
The second computes the distribution of the capacity, given that the threshold was passed. The third subsection discusses the rate at which users pass the threshold.

\subsection{Threshold Estimation}
Let $u_k$ be a threshold such that only $k$ strongest users will exceed that threshold.
% the end users estimates if their capacity belongs to top $p$ users, hence, can decide if their capacity is adequate to utilize the next slot. Hence, we would like to estimate a threshold $u_p$ such that only a fraction $p$ of all users will exceed that threshold.
Assuming that the capacity follows a Gaussian distribution $\Phi(x)$, with mean $\mu$ and variance $\sigma^2$, $u_k$ can be easily estimated using the inverse error function.
\begin{claim}\label{claim: expected capacity above est by ierf thr}
  The threshold $u_{k}$, such that $k$ users out of total $K$ users will exceed it on average is
\begin{equation}\label{eqn: p threshold expected capacity ierf}
      u_{k}  = \mu + \sigma \sqrt{2 \log \left(\frac{K}{k}\right) -\log\left[-2 \pi  \left(2 \log\left(\frac{k}{K}\right)+\log[2 \pi ]\right)\right]} + o\left(K^{-2}\right)
   \end{equation}
\end{claim}
\begin{proof}
 Let $\operatorname{erfc}^{-1}(\cdot)$ denote the complementary inverse error function. The threshold $u_k$ such that $1- \Phi(u_k) = \frac{k}{K}$ is given by
\begin{eqnarray}\label{eqn: estimated u_p Gaussian}
  %u_k &=& \mu + \sqrt{2}\sigma \operatorname{erfc}^{-1}(2p)\\
  u_k &=& \mu + \sqrt{2}\sigma \operatorname{erfc}^{-1}\left(\frac{2k}{K}\right)\nonumber\\
   &=& \mu + \sigma  \sqrt{2 \log \left(\frac{K}{k}\right) -\log\left[-2 \pi  \left(2 \log\left(\frac{k}{K}\right)+\log[2 \pi ]\right)\right]}+o\left(K^{-2}\right)\nonumber.
   %&\approx& \mu -\sqrt{2}\sigma \sqrt{\sqrt{\left(\frac{2}{\pi a}+ \frac{\log(1 - p^2)}{2}\right)^2 - \frac{\log(1-p^2)}{a}} - \left(\frac{2}{\pi a} + \frac{\log(1 - p^2)}{a} \right)}\nonumber\\
   %&\sim& \sqrt{\frac{\log(1 - p^2)}{a}}\nonumber\\
   %&=& \sqrt{\frac{\log(1 - (\frac{k}{K})^2)}{a}}
\end{eqnarray}
where the last equality follows from a Taylor series expansion.
%Substituting $u_p$ in (\ref{eqn: threshold expected capacity}), Claim~\ref{claim: expected capacity above est by ierf thr} follows.
\end{proof}

%The threshold can also be approximated for a large number of users directly, using the stability law of extreme values \cite{coles2001introduction}.

Nevertheless, using the stability law of extreme values \cite{coles2001introduction}, the threshold can also be approximated for a
large number of users directly. Indeed, using EVT, the threshold can be computed without evaluation of the inverse   $\operatorname{erfc}(\cdot)$, which cannot be evaluated in closed form. On the other hand,  the EVT relies itself on approximation. To gain sufficient amount of statistics, we logically divide the $K$  users to $\sqrt{K}$ blocks such that in each block there are $\sqrt{K}$ users, as we see in Figure~\ref{fig: GaussianExceedanceVsBlockMaxima2}. From the stability law of extreme values, the maximum in each block is still well approximated by GEV distributions. Thus, a threshold $u_{p}$, such that only a fraction  $p = \frac{k}{\sqrt{K}}$ among $\sqrt{K}$ maximal users will exceed the threshold on average, attained as follows.
\begin{claim}\label{prop: expected capacity above est by GEV thr}
The threshold $u_{p}$, such that $k$ strongest users out of total $\sqrt{K}$ strongest users will exceed it on average follows
\begin{eqnarray}\label{eqn: p threshold}
    u_{p} &=& \mu + \sigma\left(2\log \frac{\sqrt{K}}{k}\right)^{\frac{1}{2}} - \sigma\left(2\log \frac{\sqrt{K}}{k}\right)^{-\frac{1}{2}} \log \left\{ -\log \left(1 - \frac{k}{\sqrt{K}}\right)\right\} \\
    &&  + o\left(\frac{1}{\sqrt{( \log \frac{\sqrt{K}}{k})}}\right).     \nonumber
\end{eqnarray}
\end{claim}
\begin{proof}
%We can obtain the estimated threshold using EVT.
An estimated threshold can be obtained by using EVT.
 A user estimates a threshold $u_{p}$ that is near $x^F$ such that only  a fraction $p$ of the \emph{largest maximal capacities, among all maximal capacities}, will exceed.
   For all  $x$ that satisfies $a_{{p}^{-1}} x + b_{{p}^{-1}} > u_{p}$, i.e.\ are in the tail corresponding to the upper tail of Gumbel distribution, the return level $u_{p}$ is the $1 - {p}$ quantile of the Gumbel distribution for all
$0 < p < 1$, and has return period of $n = {p}^{-1}$ observations.
Thus, a user estimates the threshold by a simple quantile function,
$$1 - G_0(u_{p}) = {p}.$$
For such $u_{p}$ we have
\begin{equation}\label{eqn: last quantile threshold}
  G(u_{p}) =  \exp \{-e^{-(u_{p} - b_{{p}^{-1}})/a_{{p}^{-1}}}\} = 1 - {p}
\end{equation}
and we obtain that
\begin{equation}\label{eqn: estimated u_p Gumbel}
  u_{p} = b_{{p}^{-1}} - a_{{p}^{-1}} \log \left\{ -\log (1 - {p})\right\} + o(a_{{p}^{-1}}).
\end{equation}
The $o(a_{{p}^{-1}})$ error is derived from the Gumbel approximation error, as shown in Appendix~\ref{sec: appendix A} .
%Substituting $u_{\tilde{p}}$ in (\ref{eqn: threshold expected capacity}), we have the following.
\end{proof}

%where $a = \frac{8(\pi - 3)}{3 \pi( 4 - \pi)}$.
%which gives us similar thresholds as we can see in figure (\ref{fig: thresholdCompare}).
%Note that $u_{\tilde{p}}$ first and second terms are converging to infinity at same rate.
Note that the limit between $u_{p}$ given in (\ref{eqn: estimated u_p Gumbel}) and $u_k$ given in (\ref{eqn: estimated u_p Gaussian}) is $\frac{u_{p}}{u_k} \rightarrow \frac{3}{2\sqrt{2}} \approx 1.06$.

\begin{figure}
\centering
\subfigure[]{
\includegraphics[scale=0.28]{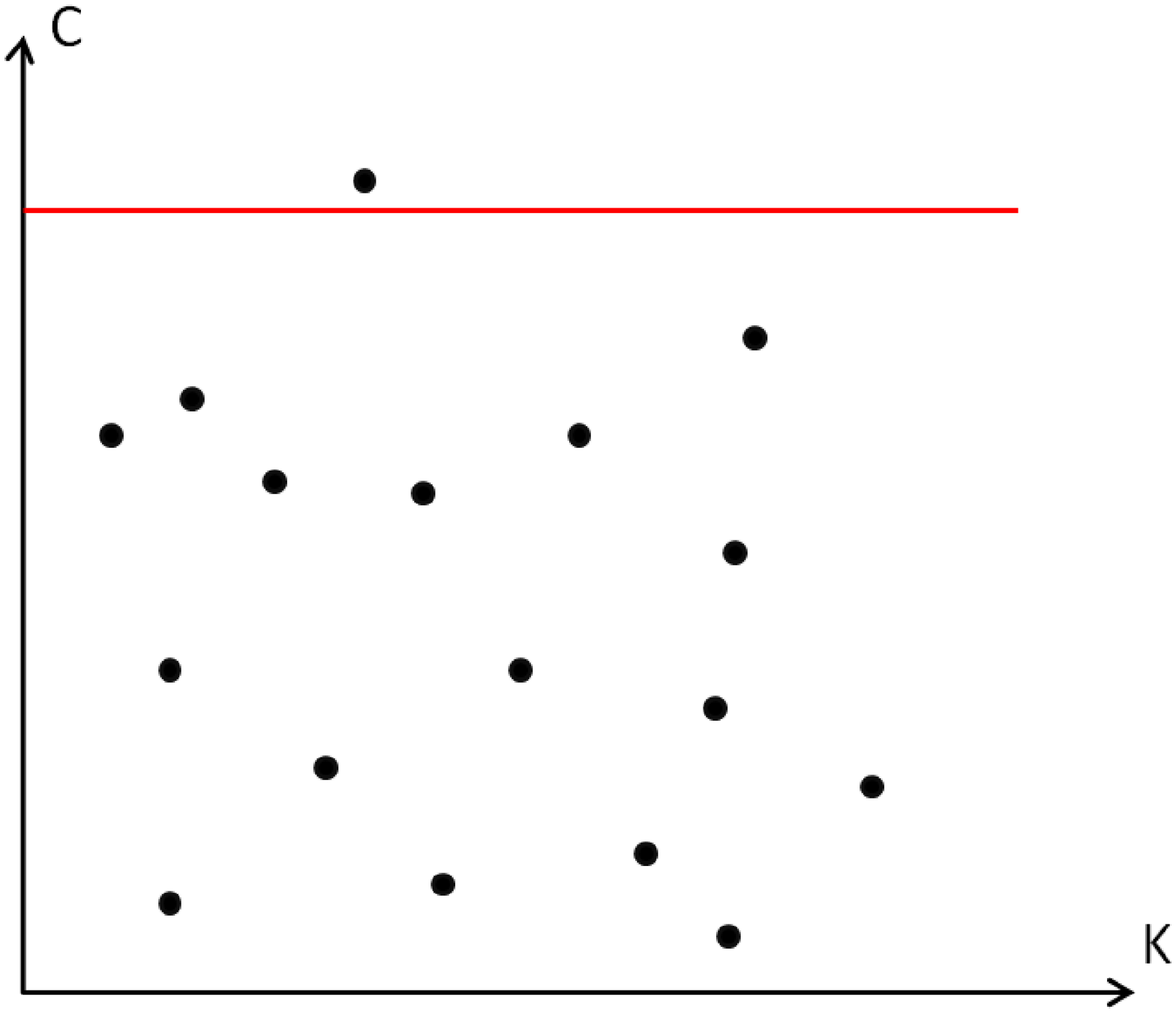}
\label{fig: GaussianExceedanceVsBlockMaxima1}
}
\subfigure[]{
\includegraphics[scale=0.28]{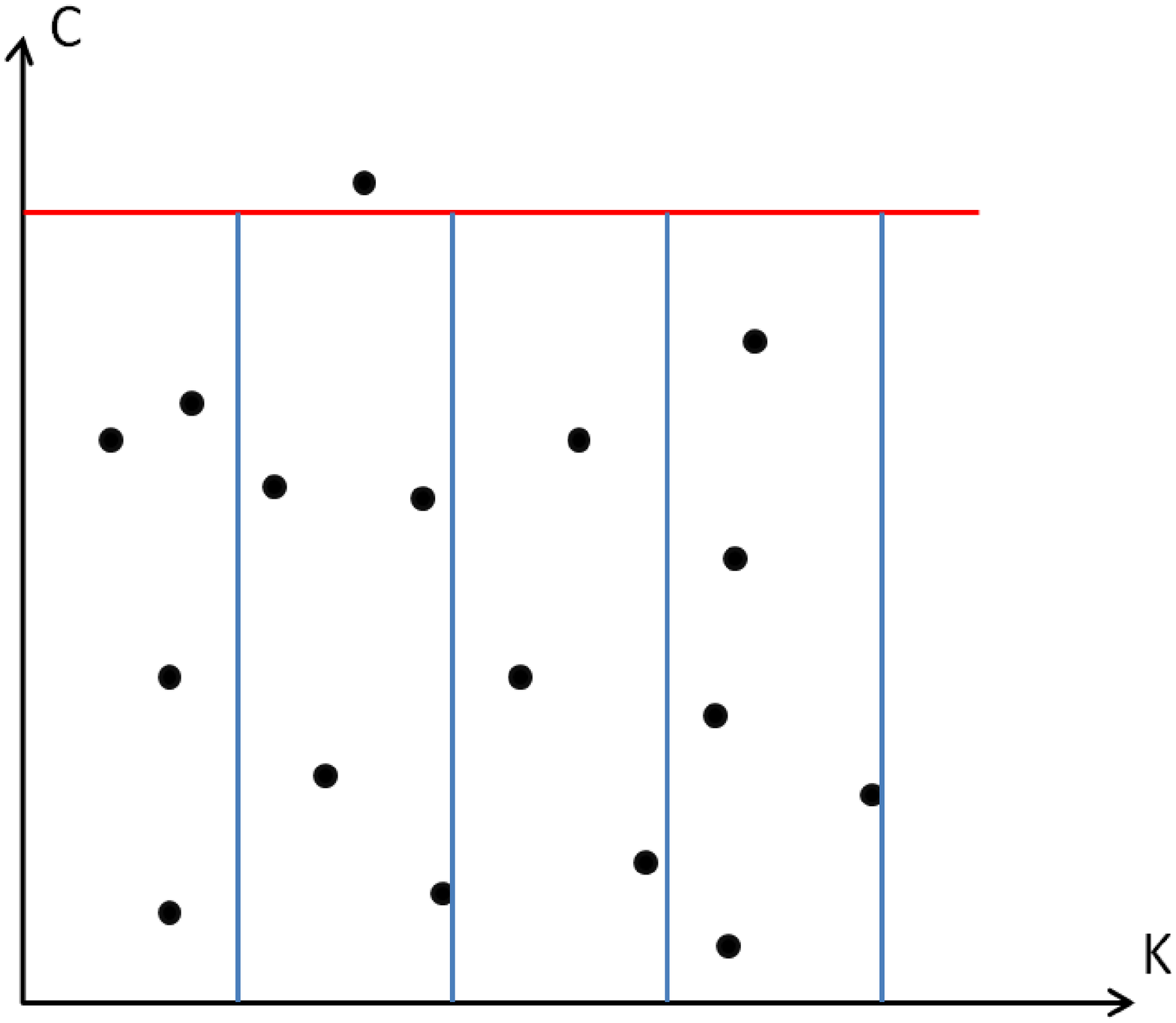}
\label{fig: GaussianExceedanceVsBlockMaxima2}
}
\caption[]{(a) $k=1$ users exceed a threshold out of $K$ observations. (b) Partitioning to $\sqrt{K}$ bins, such that in each bin there is approximately $\sqrt{K}$ users, and among this maximal users we set a threshold such that on average only the largest $k$ maximal users will exceed that threshold.}
\end{figure}

\begin{figure}
\centering
  \includegraphics[scale=1.15]{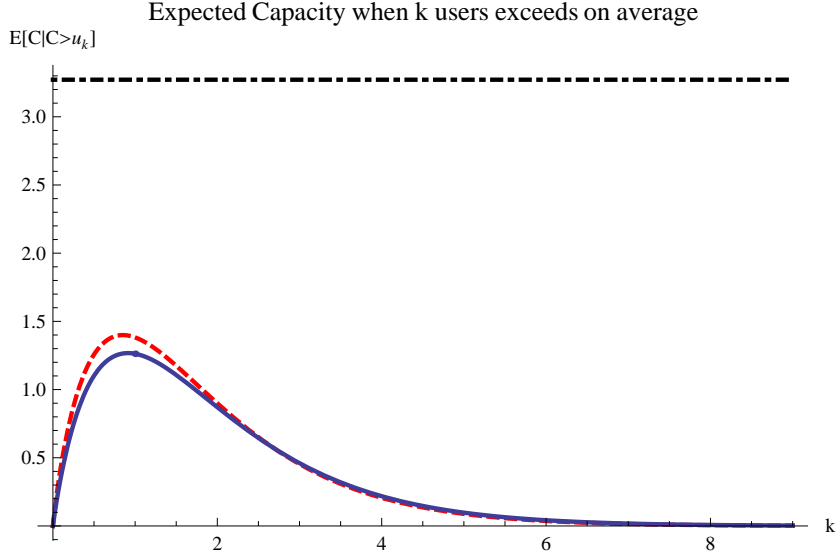}
  \caption{Threshold algorithm expected capacity gain for $K=1000$ users, when setting threshold such that $k$ users exceed on average by (\ref{eqn: estimated u_p Gaussian})(solid line) and by (\ref{eqn: estimated u_p Gumbel}) (dashed  line), comparing to the expected capacity of the optimal multi-user diversity centralized scheme (dot-dashed line).}\label{fig: thr alg unenhanced}
\end{figure}

In \cite[Proposition~4]{qin2003exploiting} it is shown that the optimal threshold (maximum throughput) is obtained by demanding that only one user exceeds on average. This is also clear from Figure~\ref{fig: thr alg unenhanced}, for both threshold estimators.

\subsection{Threshold Arrival Rate Point Process Approximation}
%\begin{proposition}\label{prop: expected capacity above thr}
%The expected capacity of all observation that exceeded threshold $u \rightarrow x^F$ is
%\begin{equation}\label{eqn: threshold expected capacity}
%    E[\mathcal{C}|\mathcal{C}> u ] = u + (2 \log K)^{-\frac{1}{2}} + o\left(\frac{1}{\sqrt{\log K}}\right).
%\end{equation}
%\end{proposition}

%We show that the behavior of large capacities is determined by $a_n, b_n$ and $\xi$, and derive a conditional model for excesses of a high threshold to obtain Proposition~\ref{prop: expected capacity above thr}.\\

%In order to prove Proposition~\ref{prop: expected capacity above thr}, we first examine the arrival rate to high thresholds, then we examine the capacity distribution given that the capacity is above threshold $u$.  We use point of process method to analyze the above, and it will help us to analyze the non-uniform case as well. Another two approaches to prove the above found in Appendix B.

In this section we  discusses the rate at which users pass the threshold. That is, for a given threshold, we examine the average number of users that exceed the threshold in a single slot.

Assume that $\textbf{x}_1,...,\textbf{x}_n$ is a sequence of $i.i.d$ random variables with a distribution function $F(x)$, such that $F(x)$ is in the domain of attraction of some GEV distribution G, with normalizing constants $a_n$ and $b_n$.\\
We construct a sequence of points $P_1,P_2,...$ on $[0,1] \times \R$ by
$$ P_n = \left\{ \left(\frac{i}{n},\frac{\textbf{x}_i - b_n}{a_n}\right) , i=1,2,...,n \right\},$$
and examine the limit process, as $n \rightarrow \infty$.

Notice that the numbers of occurrences counted in disjoint intervals are independent from each other, and large points of the process are retained in the limit process, whereas all points $x_i = o(b_n)$ can be normalized to same floor value $b_l$.
%, with
%$$b_l = \lim_{n \rightarrow \infty} \frac{x_F - b_n}{a_n}.$$
\begin{figure}[t]
\centering
\subfigure[]{
\includegraphics[scale=0.75]{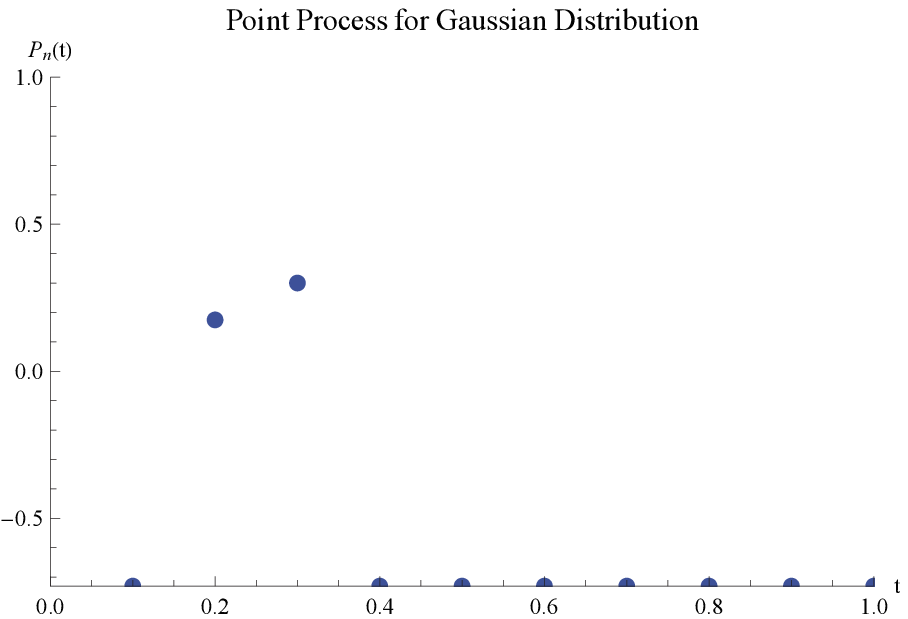}
\label{fig: point_process10}
}
\subfigure[]{
\includegraphics[scale=0.75]{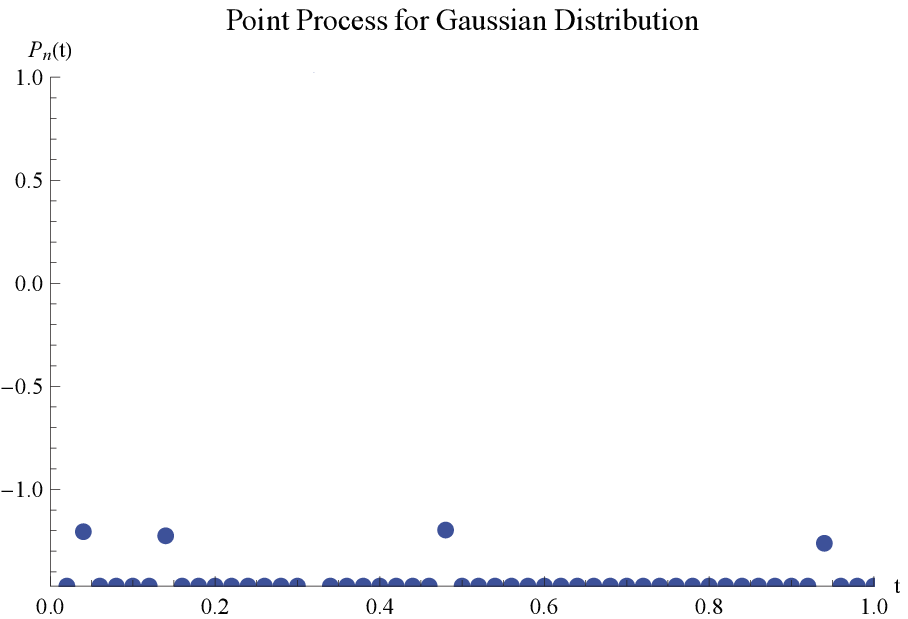}
\label{fig: point_process50}
}
\subfigure[]{
\includegraphics[scale=0.75]{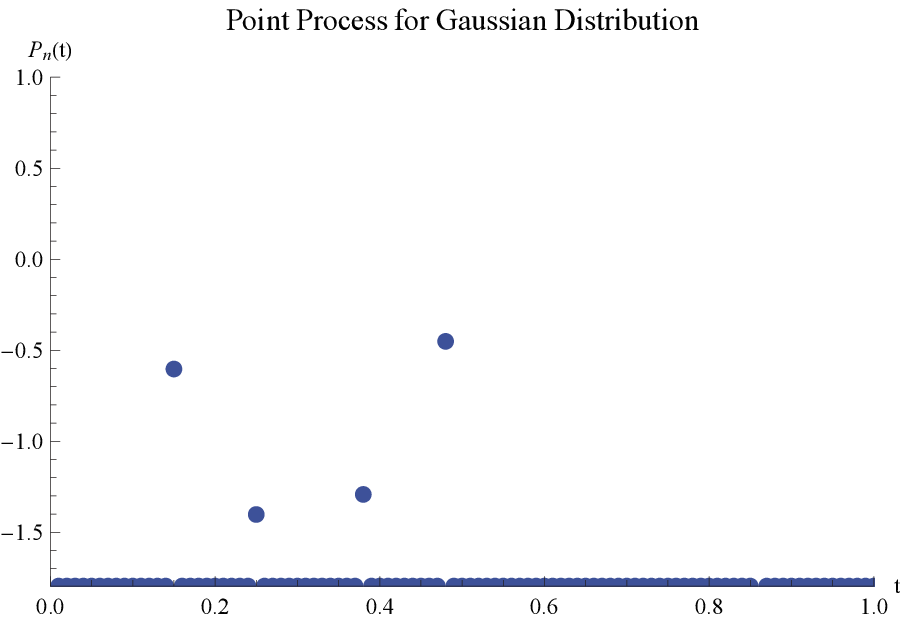}
\label{fig: point_process100}
}
\subfigure[]{
\includegraphics[scale=0.75]{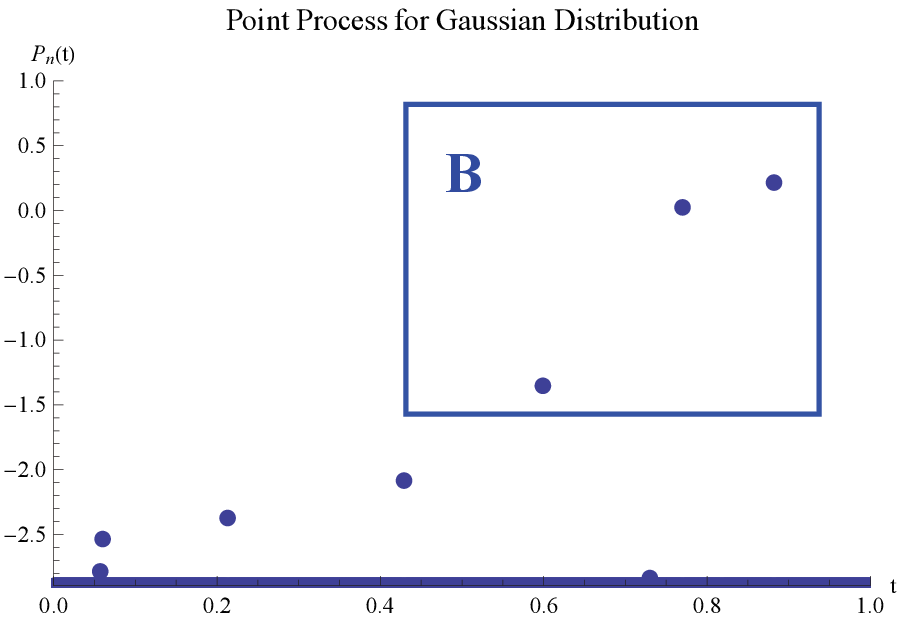}
\label{fig: point_process}
}
\caption[]{Point process for Gaussian distribution with $K\in \{10,50,100,1000\}$ users, in which all samples are normalized with $a_K,b_K$ constants. As we can see, in each of process, only a small fraction of users are above the threshold. In particular, we  obtain the expected number of arrivals to set $B$ by using (\ref{eqn: Lambda B}).  }
\end{figure}

%\begin{figure}
%  \includegraphics[scale=1.1]{point_process.eps}
%  \caption{Point process for Gaussian distribution with $K=1000$ users, in which all samples are normalized with $a_n,b_n$ constants, with floor threshold $b_l = -4 $. As we can see, only a small fraction of small users are above this threshold. In particular, we  obtain the expected number of arrivals to set $B$ by using (\ref{eqn: Lambda B}). }\label{fig: point process}
%\end{figure}

\begin{theorem}[\cite{galambos1994extreme,smith1989extreme,eastoem453}]\label{theorem: point process}
Consider $P_n$ on the set $[0,1] \times (b_l + \epsilon , \infty)$, where $\epsilon > 0$, then
$$ P_n \longrightarrow P \textmd{ as } n \rightarrow \infty $$
where $P$ is a non-homogeneous Poisson process with intensity density
$$ \lambda(t,x) = (1 + \xi x)_{+}^{-\frac{1}{\xi} - 1}$$
where $x$ is the sample value, and $t$ is the index of occurrence.
\end{theorem}
In the case where all the users are i.i.d., the process intensity density $\lambda(t,x)$ is independent in the index of occurrence $t$. For completeness, a proof in Appendix~\ref{sec: appendix C}.\\
Let $\Lambda(B)$ be the expected number of points in the set $B$. $\Lambda(B)$  can be obtained by integrating the intensity of the Poisson process over $B$, That is
\begin{equation}\label{eqn: Lambda B}
 \Lambda(B) = \int_{b \in B}{\lambda(b)db}.
\end{equation}
In this paper we are mainly interested in sets of the form
$$B_v = [0,1] \times (v,\infty)$$
where $v>b_l$. In this case
\begin{eqnarray*}
\Lambda(B_v) &=& \Lambda([0,1]\times (v,\infty))\\
&=& \int_{t=0}^1 \int_{x=v}^{\infty} \lambda(t,x)dxdt\\
&=& \int_{t=0}^1 \left[-(1+\xi x)_{+}^{-1/\xi}\right]_{x=v}^{\infty}\\
&=& \int_{t=0}^1 \left(1+\xi v\right)_{+}^{-1/\xi}dt\\
&=&(1 + \xi v)_{+}^{-1/\xi}
\end{eqnarray*}
where $a_{+}$ denotes $\max \{0,a\}$.

That is, occurrences of above threshold capacities can be modeled by a Poisson process, with parameter $\Lambda\left(B_v\right)$.
%Thus, we model exceedances of the threshold for large capacities by a Poisson process.
Namely, users normalized capacities exceed the threshold $v$ continuously and independently at a constant average rate $\Lambda(B_v)$.
In Figure~\ref{fig: point_process10}-\ref{fig: point_process} we observe the convergence of the point process to a continues process, that is, the Poisson process.
This enables us to examine important events, e.g., how likely it is to have several threshold exceedances, or what is the expected distance that users reach from the threshold. Thus, analyze the expected capacity.

%, we have the following.
%\begin{corollary}
%The length of the inter-arrival time, which is the distance between the indexes of two successive arrivals,   follows the exponential distribution with mean $\Lambda(B_u)^{-1}$.
%\end{corollary}
%\begin{corollary}
%If we use all threshold for statistical inference we have more data to use for inference but the same number of parameters to estimate as using %maxima, this suggests potential efficiency gains.
%\end{corollary}
%%%%%%%%%%%%%%%%%%%%%%
%\begin{corollary}
%Given  $k$ exceedances  $\left\{(t_j,x_j)\right\}_{j=1}^k$ with $x_j > u, \forall j$, over time period %$(0,T)$,
%then the approximating Poisson process function applied to the above exceedances of $u$ is %or Poisson %process observed on $[0,1] \times (u,\infty)$ applied to exceedances $\{(t_1,x_1),...,(t_{n_u},x_{n_u})\}$ %of $u$ is
%$$\Pr\left(N_T = k\right) = \prod_{i=1}^{r}\lambda(t_i,x_i) \exp \{ -T\Lambda([0,1] \times (u,\infty))\}.$$
%\end{corollary}

\subsection{Tail Distribution}
Focusing on points of the process $P_n$ that are above a threshold, we wish to examine the distribution of the distance that they reached from the threshold, that is the excess capacity above the threshold. %their conditional distribution given that they exceeded high threshold.\\

For any fixed $v > b_l$ let
$$u(v) = a_n v + b_n,$$
%then $u_n(v) \rightarrow x^F$,
and let $x>0$, then
%\begin{theorem}[\cite{eastoem453}]
%The distribution of points of process $P_n$ that exceeded  threshold $u_n(v)$ follows generalized Pareto distribution with parameters $\xi$
%and $\sigma_v$, where
%$$\sigma_v = 1+ \xi v$$
%\end{theorem}
%\begin{proof}
\begin{eqnarray*}
\Pr\left(\textbf{x}_i > a_n x+ u(v)| \textbf{x}_i > u(v)\right)
  &=&
\Pr\left(\frac{\textbf{x}_i - b_n}{a_n}> x + v | \frac{\textbf{x}_i - b_n}{a_n} > v\right) \\
   &=& \Pr(P_n(t) > x + v|P_n(t) > v) \\
   &\rightarrow& \Pr(P(t) > x + v|P(t) > v)
\end{eqnarray*}
where  $P_n(t)$ and $P(t)$ are the corresponding excess value $\frac{\textbf{x}_i - b_n}{a_n}$ at index $t$, and the corresponding excess value at time $t$ in the limit process, respectively. The last step is obtained from the convergence in distribution shown in  Theorem~\ref{theorem: point process}.
Now,
\begin{eqnarray}\label{eqn: Pareto def}
	  \Pr(P(t) > x + v|P(t) > v) &=& \frac{\Lambda(B_{x+v})}{\Lambda(B_v)}\\
   &=& \left[\left(1+ \xi \frac{x}{1+\xi v}\right)_{+}\right]^{-1/\xi}\nonumber\\
   &=& \left[\left(1+\xi \frac{x}{\sigma_v}\right)_{+}\right]^{-1/\xi}\nonumber
\end{eqnarray}
where $\sigma_v = 1 + \xi v $.
%\end{proof}
Hence, the limiting distribution for large threshold
$$\Pr\left(\textbf{x}_i > u(v)+a_n x| \textbf{x}_i > u(v)\right)$$
follows generalized Pareto distribution, $GPD(a_n \sigma_v,\xi)$.

%Let $\textbf{y}_{u_n(v)}$ be a non-negative random variable which represents the excess over threshold $u_n(v)$, i.e.,
%$$\textbf{y}_{u_n(v)} = (\textbf{x} - u_n(v))_{+}.$$
Note that for the Gaussian case $\xi \rightarrow 0 $, and (\ref{eqn: Pareto def}) reduces to
\begin{equation}\label{eqn: Exponential CDF}
    \Pr(\textbf{x} - u(v) \leq \alpha | \textbf{x} - u(v)>0) = 1 - e^{-\frac{\alpha}{a_n}}
\end{equation}
for all $\alpha \geq 0$.

Thus, the Gaussian distribution tail is well approximated by an exponential distribution with rate parameter $\lambda = 1/a_n$, as shown in Figure~\ref{fig: thresholdExceedance}. As a result, by taking expected value on the capacity tail distribution we obtain the corollary, which is exactly (\ref{eqn: threshold expected capacity}).
\begin{corollary}\label{coro: expected capacity above est by GEV thr}
The expected capacity seen by a user who passed the threshold $u_k$, where $k$ is the expected number of users to exceed $u_{k}$ out of $K$ users, is
\begin{eqnarray*}\label{eqn: p threshold expected capacity}
    E[\mathcal{C}|\mathcal{C}> u_{k} ] &=& u_{k} + a_K + o\left(a_{K}\right).
     %\mu +\sigma(2\log \frac{\sqrt{K}}{k})^{\frac{1}{2}} - \sigma(2\log \frac{\sqrt{K}}{k})^{-\frac{1}{2}}\\
%    \cdot \log \left\{ -\log (1 - \frac{k}{\sqrt{K}})\right\} + \sigma(2\log K)^{-\frac{1}{2}} + o\left(\frac{1}{\sqrt{( \log \frac{\sqrt{K}}{k})}}\right).     \nonumber
\end{eqnarray*}
\end{corollary}
%\begin{figure}
%, Proposition~\ref{prop: C av uniform users} follows.
%\begin{remark}
%Similarly to max-stability property, threshold stability ensures us that the limit distribution holds as long as we choose $v>b_l$.
%Hence, observing conditional distribution with large enough threshold on the limit Poisson distribution we obtain a stability of the form of the Poisson distribution.
%\end{remark}
\begin{figure}
\centering
  \includegraphics[scale=1.25]{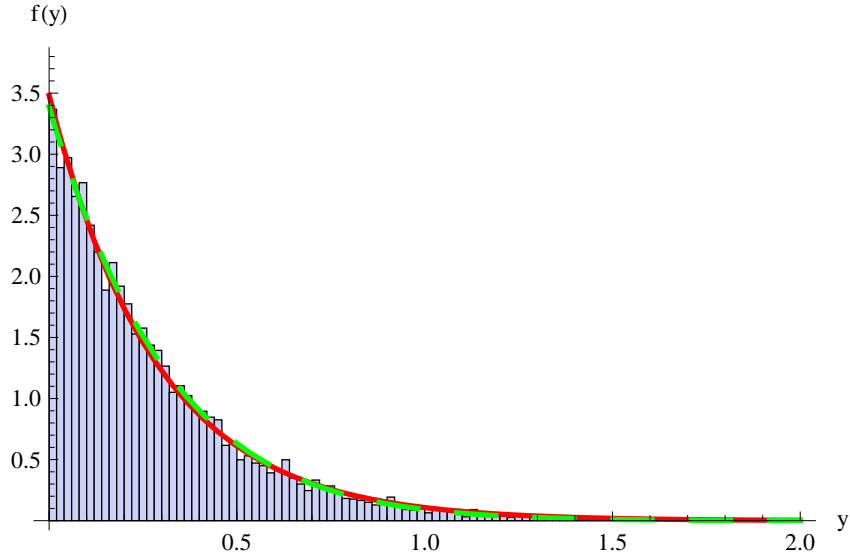}
  \caption{Tail of Gaussian distribution, statistics of 11722 observation out of 50000000 that exceed threshold 3.5, which is $ \approx 1-\Phi(3.5)$ of the observations. Dashed line is obtained by analyzing conditional distribution of Gaussian capacity given that capacity is above threshold. the solid line obtained from (\ref{eqn: Exponential CDF}). In both the threshold was derived from (\ref{eqn: estimated u_p Gumbel}) , }\label{fig: thresholdExceedance}
\end{figure}
%\end{proof}
%, hence, it describes the time between events in a Poisson process, i.e.,  arrivals to threshold $u$ occur continuously and independently at a constant average rate.
%Since the limiting distribution has exponential form, then by setting a high threshold $u$ we have the following.

\subsection{Throughput Analysis}
As mentioned, we say that a slot is utilized only if a single user transmits in that slot. If more than one user exceeds the threshold in a slot, or no user exceeds the threshold in  a slot, then the whole slot is lost. To address these scenarios we will offer a subtle collision avoidance algorithm in the following section. Using the point process method those events are very easy to analyze as we see in Figure~\ref{fig: idle-slot} and Figure~\ref{fig: collision-slot}.
%  . Similarly, we define that a time slot is idle if no user utilizing it.\\
% \cite[Proposition~4]{qin2003exploiting} Shows that under the above constraint the optimal threshold that achieves maximum throughput obtained by demanding that only one user will exceed on average. This is also clear from Figure~\ref{fig: thr alg unenhanced}, for both threshold estimators.
\begin{figure}[t]
\centering
\subfigure[]{
\includegraphics[scale=0.7]{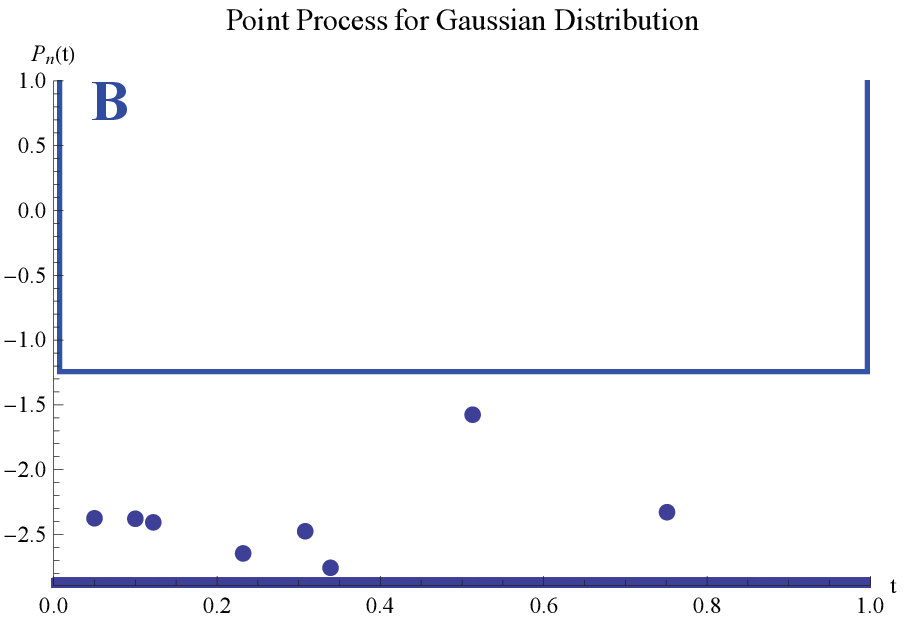}
\label{fig: idle-slot}
}
\subfigure[]{
\includegraphics[scale=0.7]{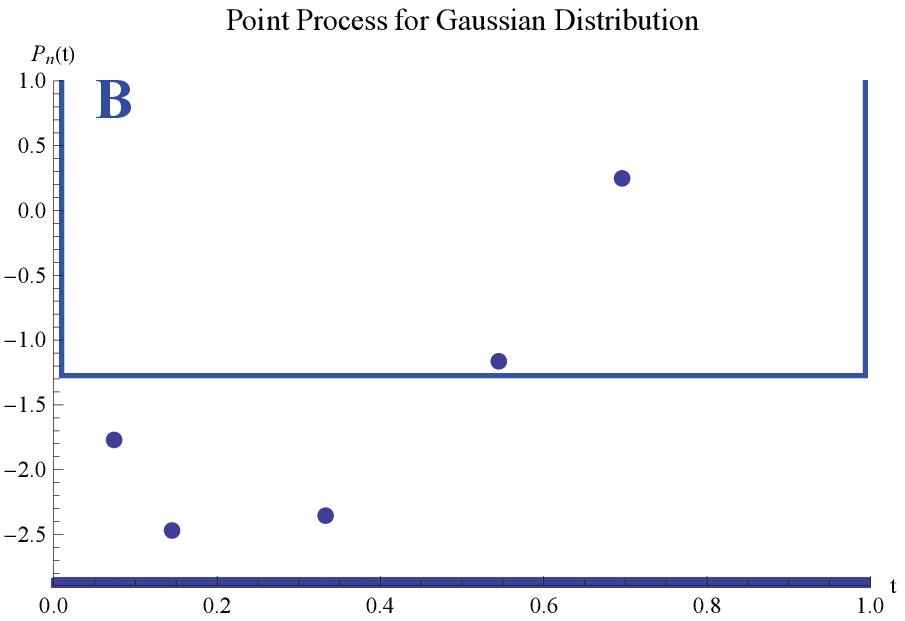}
\label{fig: collision-slot}
}
\caption[]{(a) Idle slot in point process point of view. (b) Collision slot in point process point of view. }
\end{figure}

\begin{claim}\label{claim: utilized slot probability}
For a threshold $u_k$ we have:
\begin{eqnarray}\label{eqn: utilized slot probability}
\Pr\left(\textmd{ utilized slot }\right) &=&  k e^{-k}.
\end{eqnarray}
\end{claim}
%\begin{proposition}\label{prop: optimal thr for working with single user}
%The optimal threshold $u_{p=k/K}$ for working with single user is obtained by setting  $k=1$.
%\end{proposition}
\begin{proof}
The probability that more than two out of $K$ users will exceed $u_k$ follows
\begin{eqnarray}\label{eqn: collision probability}
    \Pr( \textmd{ collision } )
    &=& \sum_{j=2}^{K}\binom{K}{j}\left(1-\Phi(u_k)\right)^{j} \left(\Phi(u_k)\right)^{K-j}\\
    &=& 1 - \left[\left(1 - \frac{k}{K}\right)^{K}  + K\left(\frac{k}{K}\right)\left(1 - \frac{k}{K}\right)^{K-1}\right]\nonumber\\
    &\stackrel{K \rightarrow \infty}{\longrightarrow}& 1 - e^{-k}(k+1)\nonumber.
\end{eqnarray}
This also implies that the number of users exceeding the threshold follows the Binomial distribution with parameters $B(K,\frac{k}{K})$, hence, converges towards the Poisson distribution as  $K$ goes to infinity.

Similarly, under the same settings, the probability of an idle slot is
\begin{eqnarray}\label{eqn: idle time slot analysis}
    \Pr\left\{\textmd{ idle slot } \right\} &=&
    \left(1 - \frac{k}{K}\right)^{K}\\
    &\stackrel{K\rightarrow \infty}{\longrightarrow}& e^{-k}\nonumber.
\end{eqnarray}
Since
$$\Pr\left(\textmd{ utilized slot }\right) = 1 - \Pr \left( \textmd{ idle slot }\bigcup \textmd{ collision } \right) $$
Claim~\ref{claim: utilized slot probability} follows.
%Taking derivative and comparing to zero yields a maximum for $k=k'$. Hence, for single user constraint, Proposition~\ref{prop: optimal thr for working with single user} follows.
\end{proof}
In particular, (\ref{eqn: idle time slot analysis}) implies that the system will be idle $e^{-1}$ of the time when setting the optimal threshold, which is a threshold such that a single user exceeds the threshold on average \cite[Proposition~4]{qin2003exploiting}.

Proposition~\ref{prop: C av uniform users} now follows from Claim~\ref{claim: utilized slot probability} and Corollary~\ref{coro: expected capacity above est by GEV thr}.

\begin{remark}
It is interesting to see that the GEV distribution given in equation (\ref{eqn: G def}) can be derived from the Point process approximation we use in this paper.

 To see this, set a threshold $u$, and for each random variable $\textbf{x}_i$ define
 $$\textbf{y}_i =  \mathbbm{1}_{\left\{\frac{\textbf{x}_i - b_n}{a_n} > u\right\}},$$
 where $\mathbbm{1}_{\{\cdot\}}$ is the indicator function.
 We have
\begin{eqnarray*}
\lim_{n \to \infty}\Pr\left\{\max_i \textbf{x}_i \leq a_n u +b_n\right\} &=& \lim_{n \to \infty}\Pr\left\{\frac{\textbf{x}_i-b_n}{a_n} \leq u \textmd{ for all }i\right\}\\
&=& \lim_{n \to \infty}\Pr\left\{\sum_i\textbf{y}_i=0\right\}\\
&=& \exp\left\{-\Lambda(B_u) \right\}\\
&=& \exp\left\{-\left(1+\xi u\right)^{-1/\xi} \right\}.
\end{eqnarray*}
The $r-largest$ users can be obtained in a similar way.
\end{remark}

%% file: hetero_trans.tex
\section{Heterogeneous Users}\label{sec. hetero}
We are now ready to address the main problem in this work. Specifically, in this section we assume that each user may be located at a different location, experiencing attenuation, delay and phase shift with \emph{different statistics} compared to other users. In our setting, the different statistics are reflected in different mean and variance of the capacity. Since users are now non-uniform, previous methods of EVT, e.g. those used in \cite{choi2008capacity}, do not apply directly. However, using the Point of Process approximation derived in the previous section with subtle modification, enable us to analyze this model and the distributed threshold scheme.

 From now on, we assume the $i$-th user capacity follows a Gaussian distribution with mean $\mu_i$ and variance $\sigma_{i}^{2}$.
Let $C_{av}^{nu}(u)$ denote the expected capacity in this non-uniform environment. Our main result is the following.
\begin{theorem}\label{theorem: expected capacity non-uniform users}
The expected capacity when working with a single user in each slot in the above non-uniform environment, where the $i_{th}$ user capacity is approximated with mean $\mu_i$ and variance $\sigma^{2}_{i}$, follows
\begin{equation}\label{eqn: threshold capacity throuput nonuniform}
 %  \left(1-\Pr(\textmd{unutilized slot})\right)E[C|C>u] =&& \\
     C_{av}^{nu}(u) =
      \frac{1}{K}\Lambda_T e^{-\frac{1}{K} \Lambda_T } \sum_{i=1}^{K}\frac{\Lambda_i}{\Lambda_T}
   \left(u + \sigma_i a_K + o(a_K)\right)\nonumber
\end{equation}
where
\begin{equation}\label{eqn: lambda i}
    \Lambda_i =  e^{-\frac{u - (\sigma_i b_K + \mu_i)}{\sigma_i a_K}}
\end{equation}
is the average threshold exceedance rate of the $i_{th}$ user, and
\begin{equation}\label{eqn: tilde lambda}
    \Lambda_T = \sum_{i=1}^K \Lambda_i
\end{equation}
is the total threshold exceedance rate.
$u$ is a threshold greater than zero that we set for all users, and $a_K,b_K$ follows (\ref{eqn: a_n normalized}) and (\ref{eqn: b_n normalized}) respectively.
\end{theorem}

%\begin{figure}
%  \includegraphics[scale=0.85]{non_uni.eps}
%  \caption{Solid line represents the expected capacity for $K=1000$ users in non-uniform environment, where the channel %capacity of each user follows Gaussian distribution with $\sigma_i \sim U[0.03,3]$ and $\mu_i \sim U[\sqrt{2}-1, \sqrt{2}+1]$, %by the analysis in Proposition~\ref{prop: expected capacity non-uniform users}. Dashed line represents  the expected capacity %when all users have the same channel capacity as the capacity of the strongest user. Dot-dashed line represents the capacity %when all users have the same channel capacity as the capacity of the mean user. }\label{fig: non uni expected capacity}
%\end{figure}
%
%\begin{figure}
%  \includegraphics[scale=0.85]{non_uni_sim.eps}
%  \caption{Solid lines represent the expected capacity analysis as in Figure~\ref{fig: non uni expected capacity}, and the %bars underneath are simulation results. }\label{fig: non uni expected capacity simulated}
%\end{figure}

Note that similar to the uniform setting,
\begin{equation}\label{eqn: C av nu definition}
C_{av}^{nu}(u) = \Pr\left(\textmd{utilized slot}\right)E[C|C>u].
\end{equation}
Thus, in this non-uniform environment as well, we first analyze the expected capacity gain when letting a user with capacity greater than the threshold to utilize a slot, and then analyze the probability that a single user utilizes a slot. Note that the computation of $C_{av}^{nu}$ is different from the uniform case, since each user channel follows a different distribution, hence, the probabilities to exceed the threshold $u$ are different. Moreover, the tail distribution the users see are different. Thus, using the point process directly in non-uniform environment will not hold.

To obtain the approximating Poisson process for this non-uniform case, we use the following method.
We build a point process for each user from his own last $M$ slots capacity value. Following Theorem~\ref{theorem: point process}, the number of threshold exceedances each user experiences, in $M$ slots which are represented in a unit interval, follows a Poisson process with rate parameter
\begin{eqnarray*}
  \Lambda_i  &=& \lim_{\xi \rightarrow 0} \left(1+\xi \frac{u - (\sigma_i b_M + \mu_i)}{\sigma_i a_M}\right)^{-\frac{1}{\xi}}\\
   &=&  e^{-\frac{u - (\sigma_i b_M + \mu_i)}{\sigma_i a_M}},
\end{eqnarray*}
where $a_M$ and $b_M$ are given in (\ref{eqn: a_n normalized}) and (\ref{eqn: b_n normalized}), respectively.
Since all users are independent, and each user exceeds the threshold according to Poisson process with rate parameter $\Lambda_i$, the total number of threshold exceedances follows a Poisson process with rate parameter
$$\Lambda_T =\sum_{i=1}^K e^{-\frac{u - (\sigma_i b_M + \mu_i)}{\sigma_i a_M}}.$$
 %Now, by consider a $1/K$ interval, we able to model the total number of the threshold exceedances for all users in a single slot.
%Formally, for a memoryless channel and sufficiently large $M$, every $M$ sample slots of the $i_{th}$ user can be represented as a point process, as shown in Theorem~\ref{theorem: point process}, with an average number of arrivals to the threshold equal to:
%\begin{eqnarray*}
%  \Lambda_i  &=& \lim_{\xi \rightarrow 0} \left(1+\xi \frac{u - (\sigma_i b_M + \mu_i)}{\sigma_i a_M}\right)^{-\frac{1}{\xi}}\\
%   &=&  e^{-\frac{u - (\sigma_i b_M + \mu_i)}{\sigma_i a_M}}
%\end{eqnarray*}
%where $a_M$ and $b_M$ are given in (\ref{eqn: a_n normalized}) and (\ref{eqn: b_n normalized}), respectively.
%Since the $K$ users are independent, and each user exceeds the threshold according to a Poisson process with average rate $\Lambda_i$,
% the total exceedance rate for $K$ users is equivalent to the sum rate of $K$ independent Poisson random variables. I.e.,
%$$\Lambda_T =\sum_{i=1}^K e^{-\frac{u - (\sigma_i b_M + \mu_i)}{\sigma_i a_M}}.$$
Now, set $M=K$ and consider a single slot interval, that is, an interval of length $1/K$ compared to the unit interval, in which the probability that a user exceed the threshold more than once is little order $o\left(1/K\right)$. Then, the total number of exceedance in this non-uniform environment follows
\begin{eqnarray}\label{eqn: non i.i.d. users point process}
    \Pr\left(\sum_{i=1}^K \textbf{N}_i=k\right) &=& \frac{\left(\frac{1}{K}\Lambda_T\right)^k}{k!}\exp\left\{-\frac{1}{K}\Lambda_T\right\}\nonumber
\end{eqnarray}
where $\textbf{N}_i$ is the number of exceedances of the $i_{th}$ user in $1/K$ time interval.
%for varying MIMO channels, the number of exceedance can be represented as the following Poisson process
%\begin{equation}\label{eqn: non i.i.d. users point process}
%    \Pr(\sum_{i=1}^K N_i(t=\frac{1}{K})=r) = \exp\{-\frac{1}{K}\sum_{i=1}^K e^{-u_p}\}\frac{(\frac{1}{K}\sum_{i=1}^K e^{-u_p})^r}{r!}.
%\end{equation}
Note that the i.i.d.\ case can be obtained by placing $\sigma_i = \sigma$ and $\mu_i = \mu, \forall i=1,2,...,K$  in (\ref{eqn: lambda i}), achieving the expression in Claim~\ref{claim: utilized slot probability}.

In order to prove Theorem~\ref{theorem: expected capacity non-uniform users}, we first prove the two claims below.
\begin{claim}\label{claim: expected capacity non iid}
Given that a single threshold exceedance occurred, then the expected capacity for non-uniform users is
\begin{equation}\label{eqn: expected thr capacity non iid users}
  E[C|C > u, \sum_{i=1}^K \textbf{N}_i=1] =
   \sum_{i=1}^{K}\frac{\Lambda_i}{\Lambda_T}
  \left(u + \sigma_i a_K + o(a_K)\right) .\nonumber
\end{equation}
\end{claim}
\begin{proof}
%Moreover, we can estimate the strongest user index vector. The probability of the $i_th$ user to be the strongest among $K$ users follows,
In the limit of each user point process, $\textbf{N}_1,\textbf{N}_2,...,\textbf{N}_K$ are independent Poisson random variables with rate parameters $\Lambda_1, \Lambda_2,...,\Lambda_K$, respectively. Thus,
%$$\left.\textbf{N}_i \right| \sum_{j=1}^{K}\textbf{N}_j = k \sim Binom\left( k , \frac{\Lambda_i}{\tilde{\Lambda}} \right).$$
the probability that only the $i_{th}$ user exceeded threshold $u$ in $1/K$ interval length is
\begin{eqnarray}\label{eqn: i user exceedance probability given the exceedance occurred not conditional}
  \Pr\left(\textbf{N}_i = 1 , \sum_{j=1}^K \textbf{N}_j = 1 \right) &=& \frac{1}{K}\Lambda_i e^{-\frac{1}{K}\Lambda_i}\prod_{j \neq i}^K e^{-\frac{1}{K}\Lambda_j}\\
  &=& \frac{1}{K}\Lambda_i e^{-\frac{1}{K}\Lambda_T}\nonumber\\
  &=& \frac{1}{K}\Lambda_T e^{-\frac{1}{K}\Lambda_T} \frac{\Lambda_i}{\Lambda_T} \nonumber.
\end{eqnarray}
Hence,
\begin{eqnarray}\label{eqn: i user exceedance probability given the exceedance occurred}
  \Pr\left(\left.\textbf{N}_i = 1 \right| \sum_{j=1}^K \textbf{N}_j = 1 \right) &=&  \frac{\Pr\left(\textbf{N}_i = 1 , \sum_{j=1}^K \textbf{N}_j = 1 \right)}{\Pr\left(\sum_{j=1}^K \textbf{N}_j = 1 \right)}\\
  &=& \frac{\Lambda_i}{\Lambda_T}.\nonumber
\end{eqnarray}
By Proposition~\ref{prop: C av uniform users}, given that the $i_{th}$ user exceeded the threshold, this user contributes $\left(u + \sigma_i a_K +  o(a_K)\right)$
to the expected capacity. By averaging user contributions, Claim~\ref{claim: expected capacity non iid} follows.
\end{proof}
\begin{claim}\label{claim: unutilized slot probability non iid}
The probability of unutilized slot for non-uniform users follows
\begin{equation}\label{eqn: unutilized slot probability non iid}
    \Pr(\textmd{ unutilized slot }) = \exp\left\{- \frac{1}{K}\Lambda_T \right\}
       + \sum_{k=2}^{K} \frac{\left(\frac{1}{K}\Lambda_{T}\right)^{k}}{k!}\exp\left\{- \frac{1}{K}\Lambda_T \right\} .\nonumber
\end{equation}
\end{claim}
\begin{proof}
The first summand is the probability of an idle slot. For non-uniform users we have
\begin{eqnarray}\label{eqn: idle slote probability non iid users}
   \Pr(\textmd{ idle slot }) &=& \Pr\left(\sum_{j=1}^{K} \textbf{N}_j =0\right)\\
    &=& e^{-\frac{1}{K}\Lambda_T}.\nonumber\\
    %&=& \exp\left\{ \frac{1}{K}\sum_{i=1}^{K} e^{\frac{u - \sigma_i b_K - \mu_i}{\sigma_i a_K}} \right\}.\nonumber
\end{eqnarray}
The second summand is the probability of collision. For this case, we have
\begin{eqnarray}\label{eqn: binomial to poisson approx for k users non iid}
    \Pr(\bigcup_{k=2}^{K} k \textmd{ users exceeds } u) &=& \sum_{k=2}^{K} \Pr\left(\sum_{j=1}^{K} \textbf{N}_j = k\right)\nonumber\\
    &=& \sum_{k=2}^{K} \frac{\left(\frac{1}{K}\Lambda_{T}\right)^{k}}{k!} e^{-\frac{1}{K}\Lambda_T} .\nonumber
   % &=& \sum_{k'=2}^{K} \exp\left\{ \frac{1}{K}\sum_{i=1}^{K} e^{\frac{u - \sigma_i b_K - \mu}{\sigma_i a_K}} \right\}\nonumber\\
   % && \frac{\left(\frac{1}{K}\sum_{i=1}^{K} e^{\frac{u - \sigma_i b_K - \mu}{\sigma_i a_K}}\right)^{k'}}{k'!}.\nonumber
\end{eqnarray}
Since
\begin{equation*}
\Pr\left(\textmd{ unutilized slot }\right) =
 \Pr\left(\textmd{ idle slot } \bigcup \textmd{ collision } \right)\nonumber
%\label{eq:}
\end{equation*}
Claim~\ref{claim: unutilized slot probability non iid} follows.
\end{proof}

\begin{figure}
\centering
\includegraphics[scale=0.6]{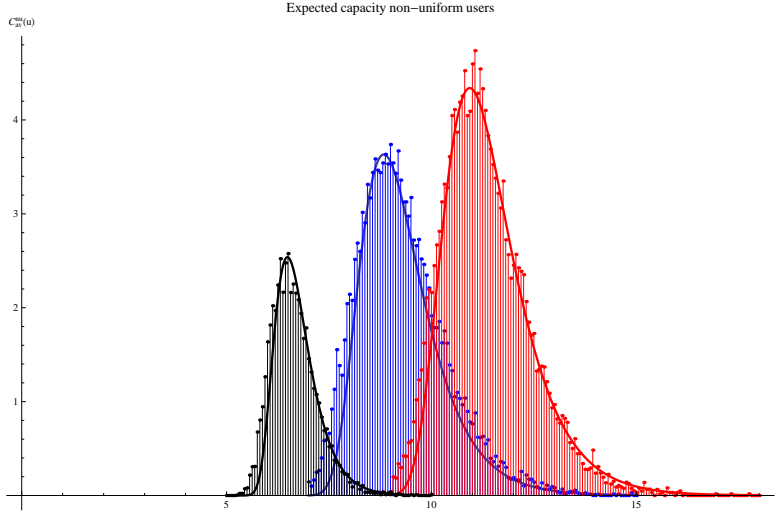}
\label{fig: non uni expected capacity simulated}
\caption{  Bars are simulation results, while the solid lines represent analytic results. The middle, blue lob, represents the expected capacity for $K=1000$ users in non-uniform environment, where the channel capacity of each user follows Gaussian distribution with $\sigma_i \sim U[0.03,3]$ and $\mu_i \sim U[\sqrt{2}-1, \sqrt{2}+1]$, by the analysis in Theorem~\ref{theorem: expected capacity non-uniform users}. Right side graph represents  the expected capacity when all users have the same channel capacity as the capacity of the strongest user. Left side graph represents the capacity when all users have the same channel capacity as the capacity of the mean user. }
\end{figure}

In Figure~\ref{fig: non uni expected capacity simulated} we present analytical results and simulated results of the expected capacity in a non-uniform environment for $K=1000$ users, and compare it to the expected capacity in a uniform environment.

\subsection{Weighted Users}
In this section, we derive the expected capacity when applying QoS to the users.
The QoS refers to communication systems that allow the transport of traffic with special requirements, e.g., media streaming, IP telephony, online games and more. In particular, a certain minimum level of bandwidth and a certain maximum latency is required to function.
 In our setting, the QoS is reflected in the exceedance probability applied to each user. This reflection allows simple analysis, which is similar to heterogeneous users analysis. Hence, given a probability vector $\vec{p} \in \mathbb{R}^{K \times 1}$, each user sets a threshold corresponding to his exceedance probability by using (\ref{eqn: estimated u_p Gaussian}) or by using (\ref{eqn: estimated u_p Gumbel}), such that his threshold arrival rate corresponds to the QoS applied to him.
Let $C_{av}^{QoS}\left(\vec{p}\right)$ denote the expected capacity in a non-uniform environment, when QoS applied to the users.% Thus, we obtain the following.
\begin{claim}\label{claim: GOS capacity}
The expected capacity with QoS in a non-uniform environment is
\begin{multline}\label{eqn: weighted threshold capacity throuput}
 %  \left(1-\Pr(\textmd{unutilized slot})\right)E\left\{C|C_i>u_i \forall_{i=1,2,...,K} \right\}= \\
 C_{av}^{QoS}\left(\vec{p}\right) = \frac{1}{K}\Lambda_{T}^{(\vec{p})} e^{-\frac{1}{K}\Lambda_{T}^{(\vec{p})}} \sum_{i=1}^{K}\frac{\Lambda_{i}^{(p_i)}}{\Lambda_{T}^{(\vec{p})}}\left(\sigma_i\left[b_{1/p_i} - a_{1/p_i}\log\log\left(1-p_i\right) + a_K\right] + \mu_i + o(a_{1/p_i})\right)\nonumber
\end{multline}
where
\begin{eqnarray}
    \Lambda_{i}^{(p_i)} &=& \exp\left\{-\frac{b_K + b_{1/p_i}}{a_K}\right\}\left(-\log (1-p_i)\right)^{a_{1/p_i}},\label{eqn: lambda p i}\\
\Lambda_{T}^{(\vec{p})} &=& \sum_{i=1}^K \Lambda_{i}^{(p_i)}\label{eqn: tilde lambda p}
\end{eqnarray}
%\begin{equation}\label{eqn: tilde lambda p}
%    \Lambda_{T}^{(\vec{p})} = \sum_{i=1}^K \Lambda_{i}^{(p_i)}
%\end{equation}
and $p_i$ is the exceedance probability of the $i_{th}$ user.
\end{claim}
Note that Claim~\ref{claim: GOS capacity} can be applied whether the users are uniformly distributed or not. That is, the QoS setting is applicable both in the previous, uniform case and in the later non-homogeneous case.
\begin{proof}
Since
$$C_{av}^{QoS} = \Pr\left(\textmd{utilized slot}\right)E\left\{C|C_i>u_i \forall_{i=1,2,...,K} \right\}$$
We analyze the following.
In (\ref{eqn: lambda i}), we expressed the threshold arrival rate as a function of the threshold $u$. Now, based on (\ref{eqn: estimated u_p Gumbel}), we wish to set a unique threshold $u_{p_i}$ for each user, such that the $i_{th}$ user will exceed his threshold with probability $p_i$. Hence,
%The rate $\Lambda_{i}^{(p_i)}$ for each user as function of the GOS that he has, by substitute threshold $u$ in (\ref{eqn: lambda i}) for the corresponding probability that given in (\ref{eqn: estimated u_p Gumbel}), and obtain the following.
\begin{eqnarray*}
    \Lambda_{i}^{(p_i)} &=& \exp\left\{-\frac{u_{ p_i} - \sigma_i b_K - \mu_i}{\sigma_i a_K}\right\}\nonumber\\
    &=& \exp \left\{\frac{-(b_{1/p_{i}} + b_K) + a_{1/p_i}\left(\log \log (1- p_i)\right) }{a_K} \right\}\nonumber\\
    &=& \exp\left\{-\frac{b_K + b_{1/p_i}}{a_K}\right\}\left(-\log (1-p_i)\right)^{a_{1/p_i}}.
\end{eqnarray*}
%Since the number of threshold exceedances of each user can be modeled as Poisson process,
 Since the users are independent, the total threshold arrival rate is the sum of rates for all users. Thus,
$$\Lambda_{T}^{(\vec{p})} = \sum_{i=1}^K \exp\left\{-\frac{(b_K + b{1/p_i})}{ a_K}\right\}\left(-\log(1-p_i)^{a_{1/p_i}}\right).$$
%Since we defined that a slot is utilized only when a single user exceed threshold, we have
As for the expected capacity, similarly to the previous section, each user that exceeds the threshold contributes a different capacity, corresponding to his threshold. Hence, by averaging the capacity that each user donates, we obtain,
\begin{eqnarray*}
E\left\{C|C_i>u_i \forall_{i=1,2,...,K}\right\} &=& \sum_{i=1}^{K}\frac{\Lambda_{i}^{(p_i)}}{\Lambda_{T}^{(\vec{p})}}\left( u_{p_i} + \sigma_i a_K + o(a_K)\right)\\
&=& \sum_{i=1}^{K}\frac{\Lambda_{i}^{(p_i)}}{\Lambda_{T}^{(\vec{p})}} \left(\sigma_i\left[b_{1/p_i} - a_{1/p_i}\log\log\left(1-p_i\right) + a_K\right] + \mu_i + o(a_{1/p_i})\right)
\end{eqnarray*}
%$$E\left\{C|C_i>u_i \forall_{i=1,2,...,K}\right\} = \sum_{i=1}^{K}\frac{\Lambda_{i}^{(p_i)}}{\Lambda_{T}^{(\vec{p})}}\left( u_{p_i} + \sigma_i a_K + o(a_K)\right)%\Pr(\textmd{the } i_{th} \textmd{user exceeded})C_{u_i}$$
% in order to obtain $E\left\{C|C_i>u_i \forall_{i=1,2,...,K}\right\}$
%where $$C_{u_i} = $$ is the capacity of the $i_{th}$ user, given that the $i_{th}$ user exceeded threshold $u_i$.
Finally, the probability that a slot is utilized, i.e., a single user exceeds the threshold in interval length of $1/K$, is
$$ \Pr\left(\textmd{utilized slot}\right) = \frac{1}{K}\Lambda_{T}^{(\vec{p})} e^{-\frac{1}{K}\Lambda_{T}^{(\vec{p})}}.$$
%Next, we substitute $u$ in (\ref{eqn: expected thr capacity non iid users}) in the estimated threshold $u_{p_i}$  obtained in (\ref{eqn: estimated u_p Gumbel}), hence, Claim~\ref{claim: GOS capacity} follows. %from the same reasons in Proposition~\ref{prop: expected capacity non-uniform users}.
Hence, Claim~\ref{claim: GOS capacity} follows.
\end{proof}

\subsection{Equal Time Sharing of Non-Uniform Users}
Equal-time-sharing is a scheduling strategy for which the system resources are equally distributed among users or groups. Whereas implementing equal-time-sharing in a homogeneous environment is to apply a uniform random or round-robin scheduling strategy to users, implementing equal-time-sharing in a non-uniform environment  is to set  for each user a threshold that is relative to his own sample maxima probability, i,e. set
$p_i = \frac{1}{K} , \forall i=1,2,...,K $.

Let $C_{av}^{es}$ denote the expected capacity in a non-uniform environment, when there is an equal exceedance probability to all users.% Thus, we obtain the following.
\begin{corollary}\label{coro: expected capacity proportional fairness}
The expected capacity with equal time sharing follows
\begin{equation}\label{eqn: expected capacity proportional fairness}
C_{av}^{es} = \frac{1}{K}\Lambda_{T}^{(1)} e^{-\frac{1}{K}\Lambda_{T}^{(1)}}
\sum_{i=1}^{K}\frac{1}{K} \left(\sigma_i\left[b_K + a_K\left(1 - \log\log\left(\frac{K-1}{K}\right) \right)\right] + \mu_i \right)
+ o(a_K)\nonumber.
\end{equation}
where
\begin{equation}\label{eqn: lambda proportianl fairness}
    \Lambda_{T}^{(1)} =  e^{-2\frac{b_K}{a_K}} K \left(-\log\left(\frac{K-1}{K}\right)\right)^{a_K}
\end{equation}
\end{corollary}
Equal time sharing is a special case of QoS. By setting $p_i = 1/K$ in (\ref{eqn: lambda p i}) to all users, Corollary~\ref{coro: expected capacity proportional fairness} follows.
%Notice that the point process is a special case where the generalized rate function is a separable function of time and space, i.e.,
%$$\lambda(x,t)=f(x)\lambda(t)$$
%for some function $f(x)$ and $\int_B f(x)dx =1$. Now, $f(x)$ represents the spatial probability density function of exceedance events in the following sense. The act of sampling this spatial Poisson process is equivalent to sampling a Poisson process with rate function $\lambda(t)$, and associating with each event a random index sampled from the probability density function $f(x)$.
%\begin{remark}
%For a  channel with memory, the average arrival rate $\Lambda_i\left(B_u,t\right)$ for each user is apparently time-dependent, but since it is a separable function of time and space, given the time dependency function $f_t$, we can derive easily the expected capacity and the utilization probability under the same technique.
%\end{remark}

%% file: capture_trans.tex
\section{Capture effect}\label{sec. capture}
Similar to the human auditory system, where the strongest speaker is filtered out of a crowed, the \emph{capture effect} is a phenomenon associated with signal reception in which in case of a collision, the stronger of two signals will be received correctly at the receiver. In this paper, the capture effect directly implies less harmful collisions, hence a higher capacity. That is, this phenomenon overcomes the situation where collisions corrupt the packets involved, and it has been shown that capture effect increase throughput and decrease delay in variety of wireless networks including radio broadcasting, such as Aloha networks, 802.11 networks, Bluetooth radios and cellular systems \cite{whitehouse2005exploiting,hadzi2002capture}.
In our settings, the capture effect enables us to set a lower threshold, such that two users will exceed the threshold on average, which significantly reduced the probability of idle slot.

 Whereas using EVT to examine the capture effect capacity gain is rather complicated, the point process technique enables us to obtain it easily.
In this section we characterize the capture effect capacity gain, when the receiver can successfully receive the transmission of the stronger user if no collision, or a collision of two users at most occurs.

\begin{proposition}\label{prop: capture effect}
The expected capacity of non-uniform users subject to capture effect follows
\begin{eqnarray}\label{eqn: capacity with capture effect}
C_{av}^{nuc}(u) &=& \frac{1}{K}\Lambda_T e^{-\frac{1}{K}\Lambda_T}\left( \sum_{i=1}^K \frac{\Lambda_i}{\Lambda_T}(u +\sigma_i a_K )\right) \\
&& + \frac{1}{2}\left(\frac{1}{K}\Lambda_T\right)^{2} e^{-\frac{1}{K}\Lambda_T} \left( \sum_{i=1}^K \sum_{j = i+1}^K  2\frac{\Lambda_i\Lambda_j}{\Lambda_T^2}  \left(u +\left(\sigma_i+\sigma_j - \frac{\sigma_i \sigma_j}{\sigma_i+\sigma_j}\right) a_K \right)\right)\nonumber\\
&& + o(a_K).\nonumber
\end{eqnarray}
where $\Lambda_i$ and $\Lambda_T$ are given in (\ref{eqn: lambda i}) and  (\ref{eqn: tilde lambda}), respectively, and $a_K$ and $b_K$ are given in (\ref{eqn: a_n normalized}) and (\ref{eqn: b_n normalized}), respectively.
\end{proposition}
\begin{proof}
%Since
%$$C_{av}^{nuc}(u) = \left(1 - \Pr\left( \textmd{unutilized slot} \right)\right)E[C|C>u]$$
The expected capacity obtained when a single user exceeds the threshold was given in Theorem~\ref{theorem: expected capacity non-uniform users}.
 To obtain the expected capacity when two users exceed threshold in a $1/K$ slot interval we define the following events:
 \begin{eqnarray*}
 A_t &=&  \textmd{exactly two users exceeded}.\\
 A_i &=& \textmd{user i exceeded}.\\
 A_j &=& \textmd{user j exceeded}.
\label{eq: events capture effect}
\end{eqnarray*}

%Hence, the expected capacity when two users exceed the threshold, subject to the capture effect follows
%%%%%%% second proof, using binomail distribution %%%%%%%%%%
%\begin{eqnarray}
%\Pr\left(A_t,A_i,A_j\right)E\left[\max\left(C_i,C_j\right)|C_i>u, C_j>u\right] \\
%= \Pr\left(A_t\right)\Pr\left(\left.A_i\right|A_t\right)\Pr\left(\left.A_j\right|A_t,A_i\right)
%E\left[\max\left(C_i,C_j\right)|C_i>u, C_j>u\right]\nonumber\\
%= \frac{\tilde{\Lambda}^2}{2}e^{-\tilde{\Lambda}}\sum_{i=1}^K \sum_{j = i+1}^K  %\left(\binom{2}{1}\frac{\Lambda_i}{\tilde{\Lambda}}\frac{\tilde{\Lambda}-\Lambda_i}{\tilde{\Lambda}}\right) %\left(\frac{\Lambda_j}{\left(\tilde{\Lambda} - %\Lambda_i\right)}\right)E\left[\max\left(C_i,C_j\right)|C_i>u, C_j>u\right]\nonumber\\
%= \frac{\tilde{\Lambda}^2}{2}e^{-\tilde{\Lambda}}\sum_{i=1}^K \sum_{j = i+1}^K  %\left(2\frac{\Lambda_i\Lambda_j}{\tilde{\Lambda}^2}\right)E\left[\max\left(C_i,C_j\right)|C_i>u, %C_j>u\right]\nonumber.
%\end{eqnarray}

%Another proof - direct way:
%\sum_{i=1}^K \sum_{j=i+1}^K
The probability that only users $i$ and $j$ exceed threshold in a $1/K$ time interval follows
\begin{eqnarray*}
\Pr(A_i,A_j,A_t)&=& \frac{1}{K}\Lambda_i e^{-\frac{1}{K}\Lambda_i} \frac{1}{K}\Lambda_j e^{-\frac{1}{K}\Lambda_j} \prod_{l \neq j,i}e^{-\frac{1}{K}\Lambda_l}\\
&=&\frac{1}{K^2}\Lambda_i \Lambda_j e^{-\frac{1}{K}\Lambda_T}\\
&=& \frac{\left(\frac{1}{K}\Lambda_T\right)^2}{2}e^{-\frac{1}{K}\Lambda_T} 2\frac{\Lambda_i \Lambda_j}{\Lambda_T^2}
\end{eqnarray*}

When two users' capacities are above the threshold, the receiver captures only the stronger user transmission, hence, only the stronger user capacity counts in practice.
The stronger user capacity distribution equals to the distribution of the maximum between two random capacities, which both have exponential tail distribution, that is, the maximum of two exponential random variables.
\begin{eqnarray*}
F_{\max\left(C_i,C_j\right)|C_i,C_j>u}(x) &=& \left(1 - e^{-\frac{x}{\sigma_i a_K}}\right)\left(1 - e^{-\frac{x}{\sigma_j a_K}}\right)\\
 &=& 1 - (e^{-\frac{x}{\sigma_i a_K}} + e^{-\frac{x}{\sigma_j a_K}}) + e^{-\frac{x}{a_K}\frac{\sigma_i \sigma_j}{\sigma_i + \sigma_j}}.
\end{eqnarray*}
Thus, when users $i$ and $j$ exceed the threshold, the stronger user will contribute%the expected capacity of a maximum between two capacities follows
\begin{eqnarray*}
u + \left(\sigma_i+\sigma_j - \frac{\sigma_i \sigma_j}{\sigma_i+\sigma_j}\right) a_K.
\end{eqnarray*}
to the expected capacity. Hence, by averaging the stronger user contribution among all $i$ and $j$, Proposition~\ref{prop: capture effect} follows.
\end{proof}

\begin{figure}
\centering
  \includegraphics[scale=1]{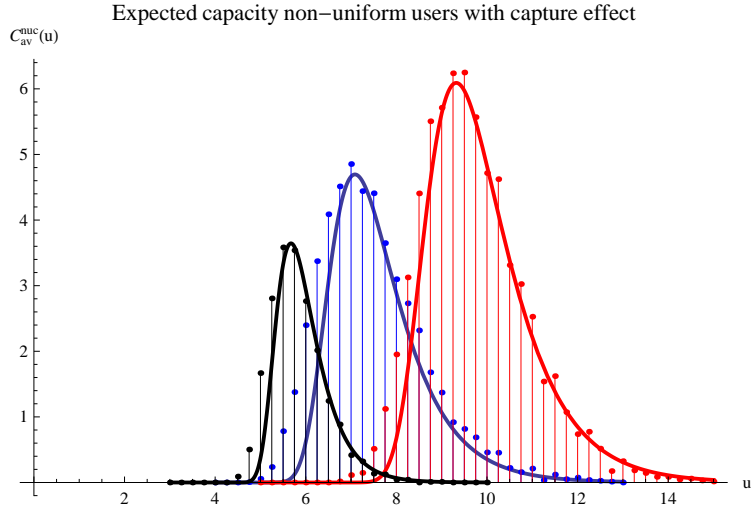}
  \caption{Expected capacity with capture effect for 250 users. On the left we present the capacity of uniform users, as if they see the same channel of the mean user, subject to capture effect. In the middle, we present the capacity of non-uniform users, subject to capture effect.  On the left, we present the capacity of uniform users, as if they see the same channel of the mean user, subject to capture effect.}
 \label{fig: capture capacity}
\end{figure}

\begin{figure}
\centering
  \includegraphics[scale=1]{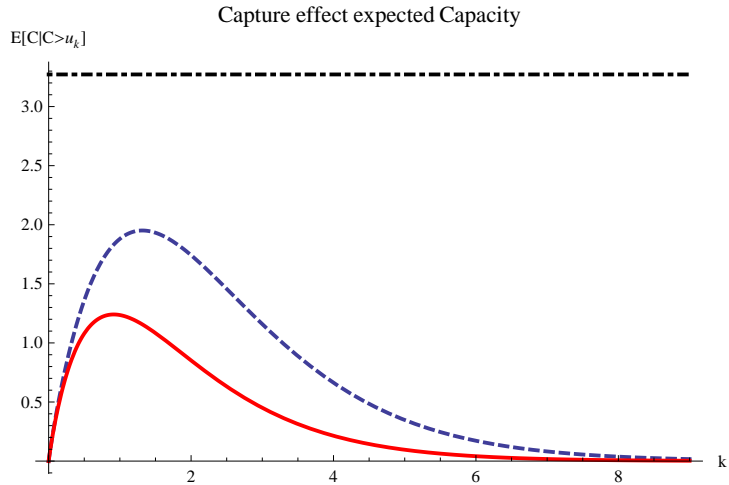}
  \caption{Capture effect capacity gain for 1000 i.i.d.\ users. The solid line represent the expected capacity when setting a threshold such that $k$ user exceeds the threshold on average, as given in Figure~\ref{fig: thr alg unenhanced}. The dashed line represent the expected capacity when $k$ users, that are subject to the capture effect, exceed the threshold on average. The upper dot-dashed line represent the expected capacity of the optimal multi-user diversity centralized scheme.}\label{fig: capture gain}
\end{figure}

In Figure~\ref{fig: capture capacity} we present the expected capacity for a uniform and non-uniform users, subject to the capture effect.
In Figure~\ref{fig: capture gain} we present the capacity gain introduced by the capture effect, when setting a threshold such that $k$ user exceeds the threshold  on average, and compare it to the expected capacity with no capture effect. Furthermore, we see that a higher capacity is achieved when setting a lower threshold, such that $k>1$ users will exceed it on average.

One should notice that the capture effect violates any QoS applied to users. When users subject to a QoS, each user must exceed a unique threshold corresponding to his QoS. Hence, when a collision occur, a strong user with a higher threshold, usually corresponding to a lower QoS,  will utilize the threshold, violating the QoS guaranteed to users with lower threshold that usually corresponds to higher QoS.
%So it seems that we always benefit from the capture effect. Though, one should notice that the capture effect violates any GOS applied to users, and the expected capacity degrades as if there is not capture effect. When users subject to a GOS, each user must exceed a unique threshold corresponds to his GOS. Hence, when a collision of two users occur, a strong user with the higher threshold (and probably with lower GOS) will utilize the threshold. And vice versa, all users with low threshold, will succeed to utilize a slot only when there is a single user which exceed threshold. hence,  the capacity will degrade to expression obtained in Proposition~\ref{prop: expected capacity non-uniform users}.
%capture effect destroy fairness in the non iid case

%% file: enhance_trans.tex
\section{Collision Avoidance}\label{sec. enhance}
In this section, we show an algorithm which asymptotically achieves the optimal capacity.
In \cite{qin2003exploiting, qin2004opportunistic}, the authors give a splitting algorithm that can cope with collisions when a collision detection mechanism is available, by dividing each slot into mini-slots, such that a collision can be resolved in the next mini-slot. In many cases, while collision resolution is not possible, the users are still capable of sensing the carrier, and understanding if a mini-slot is being used or not. Thus, we wish to develop a collision avoidance algorithm which is based only on carrier sensing. In other words, in this case we assume that the users are only able to detect if the channel is being used in mini-slots resolution. If a collision does occur within a mini-slot, we assume the whole slot is lost.
First, we wish to minimize the idle slot probability, that without any enhancement will occur $1/e$ of the time. Next, we suggest an algorithm that copes with the resulting collision probability.

From (\ref{eqn: idle time slot analysis}), it is easy to see that the idle slot probability goes to zero
when setting  $k=\log K$ as follows,
\begin{equation*}
    \Pr(\textmd{ idle time slot }) \rightarrow e^{-\log K} = 1/K \rightarrow 0.
\end{equation*}

However, when setting a threshold such that $\log K$ users will exceed on average, we have to deal with $\log K$ users on average, that find themselves adequate for utilizing next time slot.\\
To overcome this problem, we suggest to rate users that exceeded the threshold by the distance they reached from the threshold.
The set of values above the threshold is divided to $l$ bins: $[u_p,u_p+t_1),[u_p+t_1,u_p+t_2),\ldots,[u_p+t_{l-1},\infty)$, numbered $1,\ldots,l$, respectively. A user which passed the threshold checks in which bin its expected capacity lies. If the bin index is $i$, it waits $i$ mini-slots and checks the channel. If the channel is clean, it transmits its data.
%One simple way of doing this is to partition the relevant users into $l$ bins corresponding to their capacity, such that the location of each user is independent and uniformly distributed from the $l$ possibilities.\\
%Then, each user waits a time period corresponding to his bin index, and if by the time his turn has come the channel is free, he is utilizing the time slot. Which implies that if no collision occurred we expecting to achieve same performance of block maxima scheme, without exchanging CSI at all.\\
In order to achieve uniform distribution over the bins, we set the bins boundaries by the exponential limit distribution that we found in (\ref{eqn: Exponential CDF}), that is, the $i_{th}$ bin boundaries follows
\begin{equation}\label{eqn: bin boundraies}
  t_i = (2\log K)^{-1/2}\log(i/l),
   \qquad \forall i = 1,2,...,l.\nonumber
\end{equation}
as we can see in Figure~\ref{fig: usersIntoBins}.

From now on, we assume that the probability for a user who passed the threshold to fall in a specific bin is $\frac{1}{l}$ for all bins.
\begin{figure}
\centering
  \includegraphics[scale=1.15]{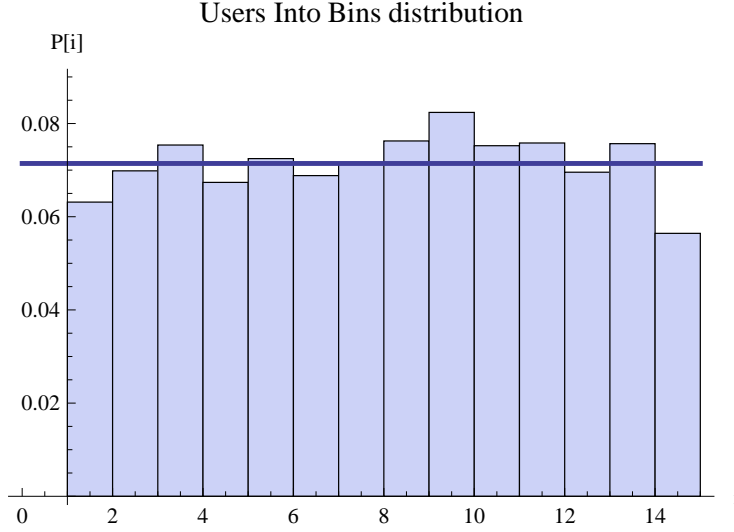}
  \caption{Distribution of users inside our bin when bin boundaries was set by (\ref{eqn: bin boundraies})  }\label{fig: usersIntoBins}
\end{figure}
\begin{claim}\label{claim: unutilized slot probability enhanced alg}
In the suggested enhanced scheme, the probability of utilized slot is
\begin{equation}\label{eqn: Collision in max index bin probability random k short}
    \Pr(\textmd{utilized slot}) =
      \sum_{j=1}^{l}\sum_{m=1}^{K}\binom{K}{m} \left(\frac{k}{K}\right)^m\left(\frac{K-k}{K}\right)^{K-m}
      m\left(\frac{1}{l}\right)\left(\frac{l-j}{l}\right)^{m-1}\nonumber
 \end{equation}
where $m$ is a realization of the number of uses who passed the threshold.
\end{claim}
\begin{proof}
Let $J$ be the index of the occupied bin with the lowest index, in which the strongest user lies. Thus, the probability that a single user occupies bin $J$, for a fixed $k$ users who exceeded the threshold is % and using union bound for a fixed $k$ users who exceeded threshold, we have,

\begin{eqnarray}\label{eqn: Collusion in max index bin probability fixed k}
    \Pr\left( \textmd{ utilized slot } \right)   &=&  \sum_{j=1}^{l} k \left(\frac{1}{ l}\right) \left(\frac{l - j}{l}\right)^{k-1}\nonumber
  % = \left(1- \frac{k}{K}  \right)^K + \nonumber\\
  %  + \sum_{j=1}^{l} \sum_{i = 2}^{k} \binom{k}{i} \left(\frac{1}{ l}\right)^i \left(\frac{l - j}{l}\right)^{k-i}\nonumber
\end{eqnarray}
We notice that when $k$ is not fixed, it should be represented as a random variable which follows the binomial distribution with parameters $n=K$ and $p=k/K$, as follows from (\ref{eqn: collision probability}). Hence, by using complete probability formula, we have
\begin{eqnarray}\label{eqn: Collision in max index bin probability random k}
    \Pr(\textmd{utilized slot}) &=&  \Pr(E_1)\\
    %\Pr(E_0) +
    %  \sum_{j=1}^{l} \sum_{m=2}^K \Pr\left(\left.E_j\right|\textbf{k} = m \right)\Pr(\textbf{k} = m) \nonumber\\
     &=&  \sum_{j=1}^{l}\sum_{m=1}^{K}\binom{K}{m} \left(\frac{k}{K}\right)^m\left(\frac{K-k}{K}\right)^{K-m}
      m\left(\frac{1}{l}\right)\left(\frac{l-j}{l}\right)^{m-1}.\nonumber
 %  &=& \left(1- \frac{k}{K}  \right)^K +\nonumber\\
 %  && +\sum_{m=2}^K\sum_{j=1}^{l} \sum_{i=2}^{m}\binom{K}{m} \left(\frac{k}{K}\right)^m \cdot \nonumber\\
 %  && \cdot \left(\frac{K-k}{K}\right)^{K-m} \cdot \nonumber\\
 %  && \cdot \binom{m}{i} \left(\frac{1}{ l}\right)^i \left(\frac{l - j}{l}\right)^{m-i}\nonumber
    %&&+\sum_{m=2}^K\sum_{j=1}^{l} \sum_{i=2}^{m}\binom{K}{m}\left(\frac{k}{K}\right)^m\left(1-\frac{k}{K}\right)^{K-m} \binom{m}{i} (1 / j)^i (1 - 1/j)^{m-i} \left(\frac{j}{l}\right)^{k}\nonumber
\end{eqnarray}
\end{proof}
This suggests that we can achieve small collision probability as we like, by increasing the number of bins, as the following claim asserts.
\begin{claim}\label{claim: number of bins needed for low unutilized probability}
 In the enhanced algorithm the probability of unutilized slot converges to zero as $l$ increases.
\end{claim}
\begin{proof}
If there are $k$ users above threshold and $l$ bins then the probability that all $k$ fall into different bins is
$$\left(1 - \frac{1}{l}\right)\cdot\left(1 - \frac{2}{l}\right)\cdot ... \cdot \left(1 - \frac{k-1}{l}\right) = \prod_{j=1}^{k-1} \left(1 - \frac{j}{l}\right)$$
Using that $1 - k/l \leq e^{-k/l}$ is tight bound when $k$ is small compared to $l$, we have
\begin{eqnarray*}\label{eqn: balls into bins}
\prod_{j=1}^{k-1} \left(1 - \frac{j}{l}\right) &\leq& \prod_{j=1}^{k-1} e^{-j/l} \\
&=& \exp \left\{-\sum_{j=1}^{k-1} \frac{j}{l}\right\}\\
&=& e^{-k(k-1)/2l}
\end{eqnarray*}
Hence, the probability of collision in any bin is $1- e^{-k(k-1)/2l}$, which is going to zero as $l$ increases.
%Since we are interested in collisions in the maximal occupied bin and the probability to fall in each bin is $1/l$, then the probability of collision in the maximal bin is
%$$\frac{1}{l}\left(1 - e^{-k(k-1)/2l} \right).$$
Hence, Claim~\ref{claim: number of bins needed for low unutilized probability} follows.
\end{proof}

\subsection{Analyzing the Delay}
Regardless of collisions that may occur, we analyze the expected time that took the maximal user decide that he is the most adequate to utilize a slot, which is equivalent to the  expected index of the maximal occupied bin, out of $l$ bins.
In order to obtain this, we order the bins in descending order, such that bin $1$ corresponds to the highest capacities.
Since we choose  $k\ll K$,  on average only a small group of users will exceed the threshold, thus, we can express the probability that bin $j$ is maximal, without using extreme distributions.\\
Let $J$ denote the index of the maximal user bin, we obtain the following.
\begin{claim}\label{claim: expected maximal bin for random k}
For a random number of users that exceeded threshold $u_p$, the expected maximal bin index $J$ follows
\begin{equation*}
  E[J] =
  \sum_{j=1}^{l}\sum_{m=1}^{K}\binom{K}{m} \left(\frac{k}{K}\right)^m\left(\frac{K-k}{K}\right)^{K-m}
  \left(\frac{l-j}{l}\right)^{m}.
% \sum_{j=1}^{l}j\frac{\left( \sum_{m=0}^K\binom{K}{m} \left(\frac{k}{K}\right)^m \left(1-\frac{k}{K}\right)^{K-m}\left(\frac{l-j+1}{l}\right)^m  \right)}{\sum_{i=1}^l \sum_{m=0}^K\binom{K}{m} \left(\frac{k}{K}\right)^m \left(1-\frac{k}{K}\right)^{K-m} \left(\frac{l-i+1}{l}\right)^m}\nonumber
\end{equation*}
%\begin{eqnarray*}
%  E[\textmd{max bin index}] &=&  \\
 %\sum_{j=1}^{l}j\frac{\left( \sum_{m=0}^K\binom{K}{m} %\left(\frac{k}{K}\right)^m %\left(1-\frac{k}{K}\right)^{K-m}\left(\frac{l-j+1}{l}\right)^m  %\right)}{\sum_{i=1}^l \sum_{m=0}^K\binom{K}{m} %\left(\frac{k}{K}\right)^m \left(1-\frac{k}{K}\right)^{K-m} %\left(\frac{l-i+1}{l}\right)^m}\nonumber
%\end{eqnarray*}
\end{claim}

\begin{figure}
  \centering
  \includegraphics[scale=1.1]{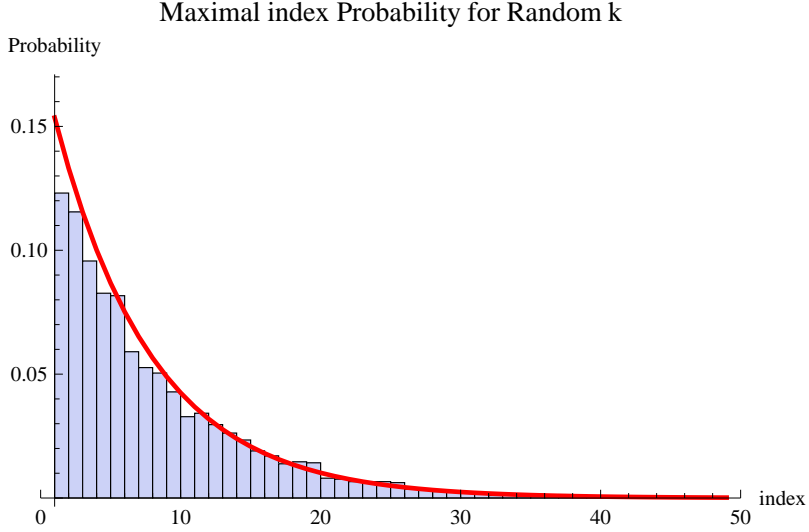}
  \caption{ maximal index simulation and analysis for random $\textbf{k}$, where the line follows (\ref{eqn: probability maximal bin random k}).}\label{fig: max_bin_random_k}
\end{figure} 

\begin{proof}
%\begin{proposition}\label{prop: expected maximal bin}
%\end{proposition}
%\begin{proof}
Given $k$ users that exceeded threshold we obtain
\begin{eqnarray}\label{eqn: expected maximal bin index fixed k}
E[J | k \textmd{ users exceeded}] &=& \sum_{j=1}^{l}\Pr(\left. J > j\right| k \textmd{ user exceeded})\\
	&=& \sum_{j=1}^{l}\left(\frac{l-j}{l}\right)^{k}.\nonumber
 % \Pr(\left. M_j \right| k \textmd{ users exceeded}) &=& \frac{\Pr(M_j)}{\sum_{i=1}^{l} \Pr(M_i)}\\
 % &=& \frac{\left(\frac{l-j+1}{l}\right)^k}{\sum_{i=1}^l \left(\frac{l-i+1}{l}\right)^k  }.\nonumber
\end{eqnarray}
%expected maximal bin index $i_1$, given that k users exceeded threshold follows
%\begin{eqnarray}\label{eqn: expected maximal bin index }
%    E[\left.\textmd{max bin index}\right|k \textmd{ users exceeded}] \\
%     = \sum_{j=1}^l j\frac{\left(\frac{l-j+1}{l}\right)^k}{\sum_{i=1}^l \left(\frac{l-i+1}{l}\right)^k  }.\nonumber
%\end{eqnarray}
By the law of total expectation we obtain the expected maximal bin index $J$, for random  $\textbf{k}$ of users  as follows.
\begin{eqnarray}\label{eqn: expected value maximal bin index from random k}
E[J] &=& E_{\textbf{k}}\left[E[J | \textbf{k} \textmd{ users exceeded}]\right]\\
  &=& \sum_{j=1}^{l}\sum_{m=1}^{K}\Pr(\textbf{k} = m)\Pr(\left. J > j\right| \textbf{k}=m)\nonumber\\
  &=& \sum_{j=1}^{l}\sum_{m=1}^{K}\binom{K}{m} \left(\frac{k}{K}\right)^m\left(\frac{K-k}{K}\right)^{K-m}
  \left(\frac{l-j}{l}\right)^{m}\nonumber.
 % \sum_{m=0}^K \Pr(\left. M_j \right| \bigcup_{i=1}^{l}M_i,\textbf{k} = m)\Pr(\textbf{k} = m)\\
 %  = \sum_{m=0}^K \frac{\Pr(\textbf{k} = m) \Pr(M_j)}{\sum_{i=1}^{l} \Pr(M_i)} \nonumber\\
 %  = \frac{\left( \sum_{m=0}^K\binom{K}{m} \left(\frac{k}{K}\right)^m \left(1-\frac{k}{K}\right)^{K-m}\left(\frac{l-j+1}{l}\right)^m  \right)}{\sum_{i=1}^l \sum_{m=0}^K\binom{K}{m} \left(\frac{k}{K}\right)^m \left(1-\frac{k}{K}\right)^{K-m} \left(\frac{l-i+1}{l}\right)^m}\nonumber
\end{eqnarray}
Hence, Claim~\ref{claim: expected maximal bin for random k} follows.
%the expected maximal bin index bin follows Proposition~\ref{prop: expected maximal bin for random k}, as we can see in Figure~\ref{fig: max_bin_random_k}.
\end{proof}
\begin{remark}
The probability that bin $J = j$ is maximal for a fixed $k$ users who exceeded the threshold follows
\begin{eqnarray}\label{eqn: probability maximal bin random k}
  \Pr(J=j) &=& \sum_{m=1}^{K}\Pr(\textbf{k}=m)\Pr\left(J=j|\textbf{k}=m\right)\\
  &=& \sum_{m=0}^{K}\binom{K}{m}\left(\frac{k}{K}\right)^{m}\left(1-\frac{k}{K}\right)^{K-m}
  \left(\left(\frac{l-j+1}{l}\right)^m-\left(\frac{l-j}{l}\right)^m\right).\nonumber
\end{eqnarray}
%Similarly, for random $\textbf{k}=m$ users who exceeded threshold we have
%\begin{eqnarray}\label{eqn: probability maximal bin fixed k}
%  \Pr(J=j) &=&  \left(\left(\frac{l-j+1}{l}\right)^m-\left(\frac{l-j}{l}\right)^m\right).
%\end{eqnarray}
\end{remark}
%\begin{figure}
%  \includegraphics[scale=0.15]{max_index_prob.eps}
%  \caption{maximal index simulation and analysis for deterministic $k$, where the blue line follows (\ref{eqn: probability maximal bin index}) vs. %maximal index simulation and analysis for random $\textbf{k}$, where the red line follows (\ref{eqn: probability maximnal bin index from random k}) %}\label{fig: max index probability}
%\end{figure}
%\begin{figure}
%\centering
%\subfigure[]{
%\includegraphics[scale=0.7]{max_bin_deterministic.eps}
%\label{fig: max bin deterministic}
%}
%\subfigure[]{
%\includegraphics[scale=0.7]{max_bin_random_k.eps}
%\label{fig: max_bin_random_k}
%}
%\caption{( maximal index simulation and analysis for random $\textbf{k}$, where the line follows (\ref{eqn: probability maximal bin random k}). }
%\end{figure}

\begin{figure}
\centering
  \includegraphics[scale=1.1]{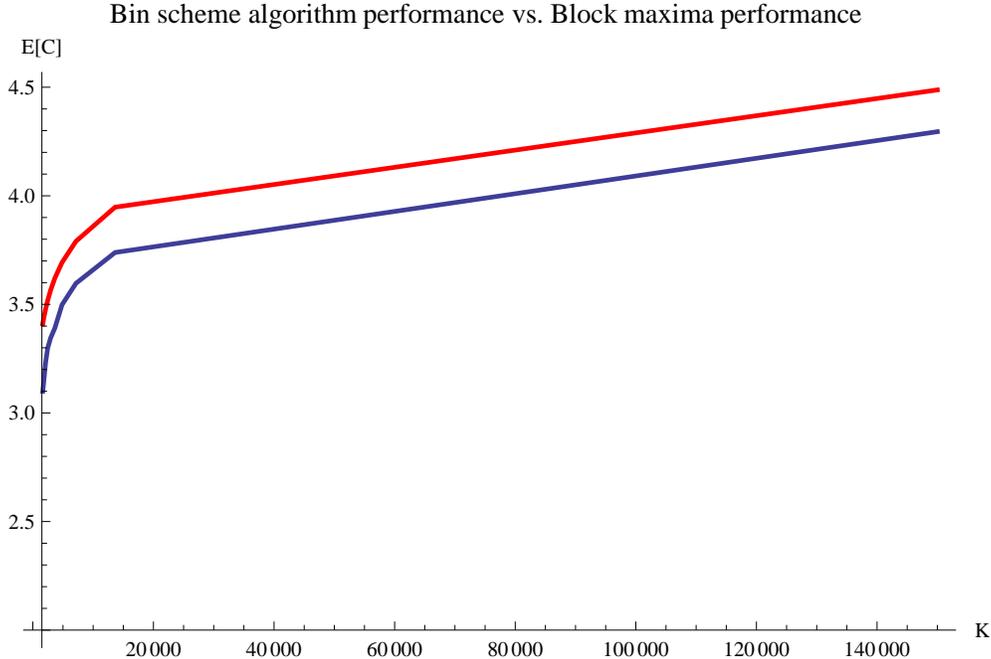}
  \caption{Bottom line - Threshold scheme expected capacity   for K users, setting threshold that on average $\lceil \log K \rceil$ users exceeds threshold, placing them into $(\lceil\log K \rceil)^2$ bins with the boundaries obtained in  (\ref{eqn: bin boundraies}).Top line is the optimal centralized scheme performance. }\label{fig: thr prf vs b maxima}
\end{figure}
In Figure~\ref{fig: thr prf vs b maxima} we see the enhanced algorithm performance when setting a threshold such that $\lceil \log K \rceil$ users exceed it on average, then placing them into $(\lceil\log K\rceil)^2$ mini-slots, comparing to the optimal centralized scheduler.

%% file: con_trans.tex
\section{Conclusion}\label{sec. conc}
In this paper, we presented a distributed scheduling scheme for exploiting multiuser diversity in a non-uniform environment, where each user has a different location, therefor will experience different channel distribution.
We characterized the scaling law of the expected capacity and the system throughput by a point process approximation, and presented a simple analysis for the expected value and throughput when applying QoS upon users. Moreover, we  presented an enhancement for the distributed algorithm in which the expected capacity and throughput reaches the optimal capacity, for  a small  delay price. 

%% file: app_A_trans.tex
\appendix

\section{Appendix A}\label{sec: appendix A}
In this section we derive the constants $a_n$ and $b_n$, for the Gaussian case.
\begin{proof}
We denote the standard normal distribution function and density function by $\Phi$ and $\phi$ respectively, and notice the relation of the tail of $\Phi$, for positive values of $x$, from Taylor series:
\begin{equation}\label{eqn: Phi relation of tail}
    1- \Phi(x)\leq \frac{\phi(x)}{x}
\end{equation}
with equality when $x \rightarrow \infty$.

First, we wish to find where $\xi$ converges to. I.e., to what distribution type the maxima of Gaussian distribution converges.
Thus, we use the relation in (\ref{eqn: Phi relation of tail}) to derive the shape parameter of Gaussian maxima,
\begin{eqnarray*}
  \xi&\approx & \frac{d}{dx}\left[ \frac{\phi(x)/x}{\phi(x)} \right]\\
   &\approx& \frac{d}{dx}\frac{1}{x}\rightarrow 0
\end{eqnarray*}
we substitute $\xi \rightarrow 0 $ in (\ref{eqn: G def}), and find the limit distribution from extreme value theory
\begin{eqnarray}\label{eqn: extreme theroy for xi=0}
  \Pr(M_n \leq u) &=& \left[\Phi(u)\right]^n \\
  &=&  \exp\left[-(1+\xi \frac{u-b_n}{a_n})\right]^{-\frac{1}{\xi}}\nonumber\\
  &\stackrel{\xi \rightarrow 0} { \longrightarrow }& \exp[-e^{-\left(\frac{u - b_n}{a_n}\right)}].\nonumber
\end{eqnarray}
That is, the maxima of Gaussian random variables converges to Gumbel distribution, where $u = a_nx +b_n$.

For retrieving the normalizing constants, $a_n$ and  $b_n$, as can be found rigorously at \cite[Theorem 1.5.3.]{leadbetter1983}, we use a well known $\log$ approximation for large values of $x$,
$$ -\log[1-(1-x)] \geq 1 - x$$
and apply it to (\ref{eqn: extreme theroy for xi=0}), i.e.,
\begin{equation}\label{eqn: log approx convergence}
    -\log [\Phi(u)]^n = -n\log (1-[1-\Phi(u)]) \geq n\left(1-\Phi(u)\right)\\
 %   &\rightarrow&  (1 + \xi x)^{-\frac{1}{\xi}} \nonumber\\
 %   & \stackrel{\xi\rightarrow 0}{\rightarrow}& \exp^{-x}\label{eq:xiConvergeToZero}
\end{equation}
hence,
\begin{equation}\label{eqn: log approx to GEV}
    n\left(1-\Phi(u)\right)\longrightarrow (1 + \xi x)^{-\frac{1}{\xi}}
\end{equation}
apply  $\xi \rightarrow 0$ to (\ref{eqn: log approx to GEV}), thus,
\begin{equation}\label{eqn: xiConvergeToZero}
    n\left(1-\Phi(u)\right)\stackrel{\xi\rightarrow 0}{\longrightarrow} e^{-x}.
\end{equation}
So, in oreder to satisfy (\ref{eqn: xiConvergeToZero}), we shell take $1 - \Phi(u) = \frac{1}{n}e^{-x}$.\\
Using again the tail relation (\ref{eqn: Phi relation of tail}), we obtain,
\begin{equation*}
    \frac{1}{n} e^{-x} \sim \frac{\phi(u)}{u}
\end{equation*}
or
\begin{equation}\label{eqn: relation of u and x}
    \frac{1}{n}e^{-x}\frac{u}{\phi(u)} \stackrel{x\rightarrow \infty}{\longrightarrow} 1
\end{equation}
applying $\log$ function on (\ref{eqn: relation of u and x}) will lead us to
\begin{equation}\label{eqn: u x log relation}
    -\log n - x + \log u - \log \phi(u) \longrightarrow 0
\end{equation}
we substitute $\phi(u)$ for a Normal density function, $\frac{1}{\sqrt{2\pi}} e^{-\frac{1}{2}u^2}$ in (\ref{eqn: u x log relation}),\\
hence,
\begin{equation}\label{eqn: u x log relation2}
    -\log n -x + \log u +\frac{1}{2}\log 2\pi + \frac{u^2}{2} \longrightarrow 0
\end{equation}
and by substitute $x = \frac{u-b_n}{a_n}$ in (\ref{eqn: u x log relation2}) and rearrange it a little, we obtain,
\begin{equation*}%\label{u x log relation3}
    -\left(\frac{u-b_n}{a_n}\right) + \log u +\frac{1}{2}\log 2\pi + \frac{u^2}{2} \longrightarrow  \log n
\end{equation*}
and since $u^2$ has the main influence on the left hand side, it implies that
\begin{equation}\label{eqn: u and n relation}
    \frac{u^2}{2\log n} \longrightarrow 1
\end{equation}
hence, by applying $\log$ to  (\ref{eqn: u and n relation}), we obtain
\begin{equation*}%\label{log u and n relation}
    2\log u - \log 2 - \log \log n \longrightarrow 0
\end{equation*}
or
\begin{equation}\label{eqn: log u def}
    \log u = \frac{1}{2}\left(\log2 +\log\log n \right) + o(1).
\end{equation}

We place (\ref{eqn: log u def}) in (\ref{eqn: u x log relation2}),and rearrange it a little to obtain
\begin{eqnarray}\label{eqn: u square}
   % u^2 = 2\left(x + \log n - \frac{1}{2} \left( \log 4\pi + \log \log n \right)\right)
   u^2 &=& 2\log n \left[  (\log n)^{-1}x + 1\right. +\\
     &&- \frac{1}{2}(\log n)^{-1}\left( \log 4\pi + \log \log n \right)+\left.o(\frac{1}{\log n})\right]\nonumber
\end{eqnarray}
and hance,
\begin{eqnarray}\label{eqn: u def}
    u &=& 2(\log n)^{\frac{1}{2}}\left[ \frac{x}{(2\log n)} + 1 + \right.\\
     &&- \left. \frac{\frac{1}{2}\left( \log 4\pi + \log \log n \right)}{(2\log n)} + o(\frac{1}{\log n})\right]\nonumber\\
    &=& (2\log n)^{-\frac{1}{2}}x + (2\log n)^{\frac{1}{2}} +\nonumber\\
    &&- \frac{1}{2}(2\log)^{-\frac{1}{2}}\left(\log \log n + \log 4\pi \right)+ o\left(\frac{1}{(\log n)^{\frac{1}{2}}}\right)\nonumber\\
    &=& a_n x + b_n + o(a_n) \nonumber
\end{eqnarray}
which means that (\ref{eqn: extreme theroy for xi=0}) follows for
\begin{equation*}%\label{eqn: a_n normalized}
    a_n = (2\log n)^{-\frac{1}{2}}
\end{equation*}
and
\begin{equation*}%\label{eqn: b_n normalized}
   b_n = (2\log n )^{\frac{1}{2}}- \frac{1}{2}(2\log n)^{-\frac{1}{2}}[\log\log n + \log(4\pi)].
\end{equation*}
%Similarly, if we substitute in (\ref{eqn: u x log relation}) $\phi(u)$ for a Gaussian density function with mean $\mu$ and variance $\sigma^2$,
%then (\ref{eqn: extreme theroy for xi=0}) follows for
%\begin{equation}\label{eqn: a_n with mu and sigma}
%   a_n = \sigma (2\log n)^{-\frac{1}{2}}
%\end{equation}
%and
%\begin{equation}\label{eqn: b_n with mu and sigma}
%   b_n = \sigma (2\log n )^{\frac{1}{2}}- \frac{1}{2}(2\log n)^{-\frac{1}{2}}[\log\log n + \log(4\pi)] + \mu.
%\end{equation}
\end{proof} 

%% file: app_C_trans.tex
\section{Appendix C}\label{sec: appendix C}

\begin{proof}(Theorem~\ref{theorem: point process})  %{ \cite[pp. 2.4]{eastoem453}}
Let $N_n(B)$ and $N(B)$ be the number of points of $P_n$ and $P$ respectively in set $B$.\\
Assuming that for any $n$ disjoint sets $B_1,B_2,...,B_n$, with $B_i\subset C, \forall i=1,2..,n$, then $N(B_1),N(B_2),...,N(B_n)$ are independent random variables.
%Then for unions of sets of the form
%$$B = (c_1,d_1]\times ... \times (c_m,d_m]$$
we will show that as $n \rightarrow \infty $
$$ E(N_n(B)) \longrightarrow E(N(B))$$
and
$$\Pr(N_n(B)=0) \longrightarrow \Pr(N(B)=0).$$
Thus, we take $B_v = (0,1] \times (v,\infty)$, such that the $i_{th}$ point of $P_n$ is in $B_v$ if
$$\frac{\textbf{x}_i - b_n}{a_n}>v$$
i.e., if $\textbf{x}_i>a_n v +b_n$.\\
The probability of this is $1 - F(a_n v + b_n)$.\\
Hence, the expected number of such points is
\begin{eqnarray*}
E[N_n(B_v)] &=& n[1-F(a_n v +b_n)]\\
&\leq& -\log\left[F(a_n v+b_n)\right]^n\\
&\rightarrow& -\log G(v)\\
&=&(1+\xi v)_{+}^{-\frac{1}{\xi}}\\
&=& \Lambda(B_v)\\
&=&E[N(B_v)].
%\label{eq:}
\end{eqnarray*}
Similarly, the event $N_n(B_v)=0$ can be expressed as
\begin{eqnarray*}
  \{N_n(B_v)\} &=& \left\{ \frac{\textbf{x}_i - b_n}{a_n}\leq v,  \forall i=1,...,n\right\} \\
   &=& \left\{\textbf{x}_i\leq a_n v +b_n \forall i = 1,...,n\right\}
\end{eqnarray*}
So
\begin{eqnarray*}
\Pr(N_n(B_v)=0) &=& \{F(a_n v + b_n)\}^n\\
&\rightarrow& G(v)\\
&=& \exp[-(1+\xi v)_{+}^{-1/\xi}]\\
&=& \exp[-\Lambda(B_v)]\\
&=& \Pr(N(B_v)=0)
%\label{eq:}
\end{eqnarray*}
\end{proof}